\documentclass[reqno]{amsart}
\usepackage{amsmath}
\usepackage{amsthm}

\usepackage{graphicx}
\usepackage{dcolumn}
\usepackage{bm}
\usepackage{orcidlink}
\usepackage[utf8]{inputenc}
\usepackage[T1]{fontenc}
\usepackage{txfonts}

\usepackage{etoolbox}

\makeatletter
\def\@email#1#2{%
 \endgroup
 \patchcmd{\titleblock@produce}
  {\frontmatter@RRAPformat}
  {\frontmatter@RRAPformat{\produce@RRAP{*#1\href{mailto:#2}{#2}}}\frontmatter@RRAPformat}
  {}{}
}%
\makeatother

\usepackage{amssymb}
\usepackage{mathtools}
\usepackage{upgreek}
\usepackage{physics}
\usepackage{enumitem}
\usepackage{tensor}
\usepackage{simpler-wick}
\usepackage[mathscr]{euscript}
\allowdisplaybreaks

\DeclareMathAlphabet{\bsf}{OT1}{cmss}{bx}{n}

\newcommand{\Ha}{\mathcal{H}}
\newcommand{\Ia}{\mathcal{I}}
\newcommand{\Oa}{\mathcal{O}}

\newcommand{\Ds}{\mathscr{D}}
\newcommand{\Is}{\mathscr{I}}
\newcommand{\Ls}{\mathscr{L}}
\newcommand{\Vs}{\mathscr{V}}

\newcommand{\sfA}{\mathsf{A}}
\newcommand{\sfD}{\mathsf{D}}

\newcommand{\sfH}{\mathsf{H}}
\newcommand{\sfM}{\mathsf{M}}
\newcommand{\sfQ}{\mathsf{Q}}
\newcommand{\sfS}{\mathsf{S}}
\newcommand{\sfk}{\mathsf{k}}
\newcommand{\sfd}{\mathsf{d}}
\newcommand{\sfe}{\mathsf{e}}

\newcommand{\ZZ}{\mathbb{Z}}
\newcommand{\CC}{\mathbb{C}}
\newcommand{\VV}{\mathbb{V}}
\newcommand{\KK}{\mathbb{K}}

\DeclareMathOperator{\res}{res}
\DeclareMathOperator{\id}{id}
\DeclareMathOperator{\End}{End}
\DeclareMathOperator{\Der}{Der}

\DeclareMathOperator{\bDer}{\overline{Der}}
\DeclareMathOperator{\nodot}{\mathop{\substack{\scriptscriptstyle\circ \\[-0.5ex] \scriptscriptstyle\circ}}}
\DeclareMathOperator{\reg}{\texttt{reg}}
\DeclareMathOperator{\Lie}{Lie}

\DeclarePairedDelimiter\floor{\lfloor}{\rfloor}

\newcommand{\g}{\mathfrak g}
\newcommand{\h}{\mathfrak h}

\newcommand{\n}{\mathfrak n}
\renewcommand{\sl}{\mathfrak{sl}}
\newcommand{\gl}{\mathfrak{gl}}
\newcommand{\hsl}{\widehat{\mathfrak{sl}}}
\newcommand{\og}{\mathring{\g}}
\newcommand{\oh}{\mathring{\h}}
\newcommand{\on}{\mathring{\n}}
\newcommand{\upcirc}[1]{\mathring{#1}}
\newcommand{\ff}[3]{\tensor{f}{_#1_#2^#3}\,}
\newcommand{\lp}{[t,t\inv]}
\newcommand{\no}[1]{\nodot #1 \nodot}
\newcommand{\npr}[2]{{(#2)}_{(#1)}}

\newcommand{\gam}{\upgamma}
\newcommand{\bet}{\upbeta}
\newcommand{\tet}{\uptheta}
\newcommand{\ro}{\uprho}
\newcommand{\jota}{\jmath}

\newcommand{\vac}{\ket{0}}
\newcommand{\vak}{\ket{0}^k}

\newcommand{\inv}{^{-1}}
\newcommand{\ie}{\emph{i.e. }}
\newcommand{\wt}{\widetilde}
\newcommand{\bM}{\overline{\sfM}}
\newcommand{\bMM}{\overline{\bsf{M}}}
\newcommand{\tM}{\widetilde{\sfM}}

\newcommand{\hM}{\widehat{\sfM}}
\newcommand{\hMn}[1]{{\widehat{\sfM}_{z}}[\leq \! #1]}
\newcommand{\bMn}[1]{{\overline{\sfM}_{z}}[\leq \! #1]}
\newcommand{\tMn}[1]{{\widetilde{\sfM}_{z}}[\leq \! #1]}
\newcommand{\bMz}{\bM_z}
\newcommand{\nn}{\nonumber}
\newtheorem{theorem}{Theorem}[section]
\newtheorem*{theorem*}{Theorem}
\newtheorem{proposition}{Proposition}[section]

\newtheorem{definition}{Definition}[section]
\newtheorem{lemma}{Lemma}[section]
\theoremstyle{definition}

\newtheorem*{remark*}{Remark}

\usepackage{hyperref,cleveref}

\usepackage{dcolumn}

\begin{document}

\title[]{Wakimoto construction for double loop algebras\\and $\zeta$-function regularisation}

\author{Tommaso Franzini\orcidlink{https://orcid.org/0000-0002-4040-0063}}
\email{tommaso.franzini@gmail.com}
\address{Department of Physics, Astronomy and Mathematics, University of Hertfordshire, Hatfield, AL10 9AB, United Kingdom}
\address{School of Mathematics and Hamilton Mathematics Institute,
  Trinity College Dublin, Ireland}

\date{\today}

\begin{abstract}
The Feigin-Frenkel homomorphism underpinning the Wakimoto construction realises an affine Lie algebra at critical level in terms of the $\bet\gam$-system of free fields. It was recently shown that much of the construction also goes through for double loop algebras. However, certain divergent sums appear. In this paper we show that, suggestively, these sums vanish when one performs $\zeta$-function regularisation.
\end{abstract}

\maketitle

\section{Introduction}
Free field realisations are a powerful tool in chiral conformal field theory; they embed intricate symmetry algebras into simpler algebras of the modes of free fields.
An important example is the free field realisation of an affine Lie algebra in terms of a Heisenberg algebra, first proposed by Wakimoto \cite{Wak} in the case of $\mathfrak{sl}_2$, and later generalised to any affine Lie algebra by Feigin and Frenkel in the untwisted case \cite{FF} and by Szczesny in the twisted one \cite{szczesny2001wakimoto}. This construction can be used to define Wakimoto modules, which are of central importance in the study of integrable systems as they play a fundamental role in the description of the spectra of quantum Gaudin models of finite type \cite{FFR} and in the computation of correlation functions in the WZW model \cite{FFboson}.

For any simple Lie algebra $\g$ of finite type, this realisation can be mathematically formulated as a vertex algebra homomorphism between the vacuum Verma module at critical level over the corresponding untwisted affine algebra $\widehat\g$ and the Fock space for the $\bet\gam$-system of free fields \cite{frenkelvertex,frenkelloop}, as we will now recall.

\subsection{Free field realisation for finite-type Lie algebras}

For any finite-dimensional simple Lie algebra $\g=\n_-\oplus \h \oplus \n_+$ there is a Lie algebra homomorphism
\begin{align}\label{Lieahom}
  \rho: \g \to \Der\Oa(\n_+)
\end{align}
realising the Lie algebra as differential operators on the algebra of polynomial functions $\Oa(\n_+)$ on the unipotent group $U=\exp(\n_+)\simeq \n_+$. Choosing coordinates $X^{\alpha}$ labelled by the positive roots, $\alpha\in\Delta_+$, then $\Oa(\n_+)=\CC[X^{\alpha}]_{\alpha\in\Delta_+}$.
For example, for $\sl_2$ one has the well-known realisation
\begin{align}\label{sl2ex}
  E \mapsto D \qquad H\mapsto -2X D \qquad F\mapsto -X^2D,
\end{align}
where $E,F$ and $H$ are the Chevalley-Serre generators and $X$, $D:=\partial_{X}$ are the generator of the Weyl algebra with commutation relations $[D,X] = 1$.

The main fact underpinning the Wakimoto construction is that this homomorphism can be promoted to a homomorphism of vertex algebras, as follows.
First, consider the affine algebra
\begin{equation}
  \label{eq:affalg}
  \widehat{\g}\simeq_{\CC}\g[t,t\inv]\oplus \CC\sfk \oplus \CC \sfd,
\end{equation}
where $\sfk$ is the central element and $\sfd$ is the derivation in the homogeneous gradation, \ie $\sfd=t\partial_t$.
It has $\widehat{\g}_+:=\g\otimes\CC[t]\oplus\CC \sfk$ as a subalgebra. We define $\CC\vac^k$ the one-dimensional representation of $\widehat{\g}_+$ in which $\g[t]$ acts trivially and the central element acts as multiplication by $k\in\CC$, called the \emph{level} of the representation. The vacuum Verma module at level $k$ is the induced module
\begin{equation}
  \label{eq:1}
  \VV_0^{\g,k} = U(\widehat{\g})\otimes_{U(\widehat{\g}_+)}\CC\vac^k
\end{equation}
This space has the structure of a vertex algebra (cfr. \cite[Theorem 2.4.5]{frenkelvertex}).

On the other hand, one can define the Fock module for the $\bet\gam$-system on $\n_+$, $\sfM(\n_+)$ (see \cref{sec:Heis} for the details on its construction).
Denoting with $\Omega_{\Oa(\n_+)}$ the space of 1-forms, one has $\widehat{\g} \simeq \VV_0^{\g,k}[1]$, $\Oa(\n_+)\simeq\sfM(\n_+)[0]$ and $\Der\Oa(\n_+)\oplus \Omega_{\Oa(\n_+)}\simeq \sfM(\n_+)[1]$. It is possible to construct the following map of vector spaces
\begin{align}
  \VV_0^{\g,k} [\leq 1] \to \sfM(\n_+)[\leq 1],
\end{align}
identifying the corresponding vacuum vectors and mapping
\begin{align}
  A[-1]\vak \mapsto \sum_{\alpha\in\Delta_+}P_A^{\alpha}(\gam[0])\bet_{\alpha}[-1]\vac.
\end{align}
Crucially, this map is not a map of vertex algebras, as it does not preserve non-negative products.
This problem can be solved by adding extra terms, as shown in \cite{FF,frenkelloop}. Formally, there exists a linear map
\begin{align}\label{eq:phimap}
  \phi: \g \to \Omega_{\Oa(\n_+)},
\end{align}
such that
\begin{align}
  \rho + \phi : \g \to \Der\Oa(\n_+)\oplus \Omega_{\Oa(\n_+)}
\end{align}
can be lifted to a map of vertex algebras,
\begin{align}
   \theta:\VV_0^{\g,-h^{\vee}} \to \sfM(\n_+),
\end{align}
where, crucially, the level has to be set to the critical level $-h^{\vee}$, $h^{\vee}$ being the dual Coxeter number of $\g$.

In the case of $\sl_2$, \cref{sl2ex} is lifted to the vertex algebra map
\begin{equation}
  \begin{gathered}
    E[-1]\vac \mapsto \bet[-1]\vac \qquad H[-1]\vac \mapsto -2\gam[0]\bet[-1]\vac\\
    F[-1]\vac \mapsto -\gam[0]\gam[0]\bet[-1]\vac-2\gam[-1]\vac.
  \end{gathered}
\end{equation}

\subsection{Generalisation to the affine case}\label{sec:homLA}
It is natural to ask whether this construction generalises to the case where $\g$ itself is an untwisted affine algebra.
Perhaps surprisingly, as shown in \cite{analogFFR}, much of it does, as follows.

The algebra $\g$ still admits a triangular decomposition, where the various subalgebras are now infinite-dimensional and in general not nilpotent, namely $\g\cong \n_-\oplus \h \oplus \n_+$. We denote by $\og\cong \on_-\oplus\oh\oplus\on_+$ the corresponding underlying finite-type Lie algebra, with $\upcirc\Delta$ its space of roots.
A basis of $\g$ is given by
\begin{align}
    \{\sfk,\sfd\}\cup \{J_{a,n} \}_{a\in \Ia;n\in\ZZ},
\end{align}
\sloppy where $J_{a,n}:=J_a\otimes t^n$ and $\Ia:=(\upcirc\Delta\setminus\{0\})\cup \{1,\dots,\rank \og\}$. We define the index set $\sfA:=\{(a,0)\}_{a\in \upcirc\Delta_+}\cup(\Ia\times\ZZ_{\geq 1})$, so that a basis for $\n_+$ is given by $\{J_{a,n}\}_{(a,n)\in\sfA}$.

Following \cite{Kum}, one can define the pro-nilpotent algebra $\wt{\n}_+=\prod_{\alpha\in\Delta_+}\g_{\alpha}$.
More explicitly, this space can be defined as the inverse limit of  a system of nilpotent Lie algebras $\n_+/\n_{\geq k}$, where
\begin{equation}
  \label{eq:3}
  \n_{\geq k}=\bigoplus_{\substack{\alpha\in\Delta_+\\\text{ht}(\alpha)\geq k}} \g_{\alpha}
\end{equation}
is the ideal containing elements with degree higher than $k$. Here $\text{ht}$ denotes the grade of a root in the homogeneous gradation, \ie $\text{ht}(\delta n+\alpha)=n$, for $n\in\ZZ$, where $\delta\in\h^*$ is the imaginary root, defined as the sum of all roots multiplied by their corresponding Dynkin label.
Therefore, elements of this completion are possibly infinite sums of the form $\sum_{\alpha\in\Delta_+} x^{\alpha}$, with $x^{\alpha}\in\g_{\alpha}$, provided they truncate to finite sum modulo $\n_{\geq k}$ for any $k$.

Via the exponential map, one can similarly define the pro-group $U$ as the inverse limit of the groups $\exp(\n_+/\n_{\geq k})$, where the multiplicative structure is given by the Baker-Campbell-Hausdorff formula and whose elements are infinite products of the form $\prod_{(a,n)\in\sfA}\exp(x^{a,n}J_{a,n})$ with $x^{a,n}\in\CC$, provided they truncate to finite ones modulo terms $\exp(\n_+/\n_{\geq k})$ for any $k\in\ZZ_{\geq 0}$.

One introduces a set of coordinates $X^{a,n}:U\to\CC$ on $U$, such that for $g=\prod_{(a,n)\in\sfA}\exp(x^{a,n}J_{a,n})$, $X^{a,n}(g) =x^{a,n}$, which define the $\CC$-algebra of polynomial functions on $U$,
\begin{align}
    \Oa(\n_+)=\CC[X^{a,n}]_{(a,n)\in\sfA}.
\end{align}

In this setting, the Weyl algebra $\Ha$ is the free unital $\CC$-algebra generated by $X^{a,n}$ and $D_{a,n}$, quotiented by the relations
\begin{align}
  [X^{a,n},X^{b,m}]=0=[D_{a,n},D_{b,m}], \qquad [D_{a,m},X^{b,n}]=\delta^b_a\delta^n_m.
\end{align}

The space of derivations on $\Oa(\n_+)$, $\Der \Oa(\n_+)$, is a subalgebra of the Weyl algebra $\Ha$ with elements of the form $\sum_{(a,n)\in\sfA}P^{a,n}(X)D_{a,n}$, where $P^{a,n}(X)\in\Oa(\n_+)$, where only a finite number of terms is non-zero; the respective completions $\wt{\Ha}$ and $\wt{\Der}\Oa(\n_+)$, are those algebras where this last restriction is lifted.

There is a continuous homomorphism of Lie algebras
  \begin{align}
  \label{eq:liehom1}
      \varrho : \g &\longrightarrow \wt{\Der} \Oa(\n_+)\nn\\
      A &\longmapsto \sum_{(a,n)\in\sfA}P^{a,n}_A(X)D_{a,n}
  \end{align}

The case of $\hsl_2$ has been explicitly worked out in \cite{analogFFR}; as a novel example, consider the Cartan-Weyl basis for $\sl_3$ given by $\{E_{\pm\alpha}\}_{\alpha\in\upcirc{\Delta}_+}$ where $\upcirc\Delta_+=\{\alpha_1, \alpha_2,\alpha_1+\alpha_2\}$, together with the Cartan generators $\{H_i\}_{i=1,2}$. An explicit matrix representation is given in terms of $3\times 3$ matrices, by the identification $E_{\alpha_1}\mapsto \mathsf{e}_{12}$, $E_{\alpha_{2}}\mapsto \sfe_{23}$, $E_{\alpha_1+\alpha_2}\mapsto \sfe_{13}$, $E_{-\alpha_1}\mapsto \sfe_{21}$, $E_{-\alpha_2}\mapsto \sfe_{32}$, $E_{-\alpha_1-\alpha_2}\mapsto \sfe_{31}$ and $H_{1}\mapsto \sfe_{11}-\sfe_{22}$, $H_{2}\mapsto \sfe_{22}-\sfe_{33}$, where $\sfe_{ij}$ is the matrix with a 1 in position $(i,j)$ and zero elsewhere.

One finds for example
\begin{align}\label{E11}
    \varrho(J_{\alpha_1,1}) =& D_{{\alpha_1},1} -\sum_{k\geq 2}X^{{\alpha_2},k-1}D_{{\alpha_1+\alpha_2},k} - \sum_{k\geq 3}X^{{-\alpha_1},k-1}D_{1,k} \nn\\
    &+ \sum_{k\geq 3}X^{{-\alpha_1-\alpha_2},k-1}D_{{\alpha_2},k} +2 \sum_{k\geq 3}X^{{1},k-1}D_{{\alpha_1},k} - \sum_{k\geq 3}X^{{2},k-1}D_{{\alpha_1}\,k}\\
    &+(- X^{{-\alpha_2},2}X^{{\alpha_2},1}+\dots) D_{{\alpha_1},4} +(- X^{{-\alpha_1-\alpha_2},2}X^{{\alpha_2},1}+\dots) D_{1,4}\nn\\
      & + (- X^{{-\alpha_1-\alpha_2},2}X^{{\alpha_2},1}+\dots)D_{2,4} +( X^{1,2}X^{{\alpha_2},1} + \dots) D_{{\alpha_1+\alpha_2},4}+\dots\nn
\end{align}

In the appendix we present more examples of this realisation.

\subsection{The widening gap subalgebra}\label{sec:WGsub}
Let us now briefly comment on the general structure of the terms in \cref{E11}. For each element $J_{a,n}\in\g$ one finds that
\begin{align}
    \varrho(J_{a,n}) = f\indices{_a_b^c} \sum_{k\geq N} X^{b,k-n} D_{c,k} + \sum_{(b,m)\in\sfA}R_{J_{a,n}}^{b,m}(X)D_{b,m}.
\end{align}
for some $N\in\ZZ_{\geq 0}$ depending on $a,b,c$ and $n$. Here, the first term is a quadratic infinite sum in the generators and will be of central importance in the sections below. Indeed, this kind of sums will be the main source of problems when trying to lift the homomorphism of Lie algebras to one of vertex algebras. The second term is part of a subalgebra of $\wt{\Der}\Oa(\n_+)$, of derivations of $\Oa(\n_+)$ with \emph{widening gap}, $\bDer\Oa(\n_+)$. We will give a precise definition in \cref{sec:WG}, but roughly speaking these are those possibly infinite sums where the loop degree of each $X$ factor in the polynomial $R$ grows ``slower'' than the loop degree of the corresponding $D$, creating a gap between them that eventually widens.
For example in the case of $\widehat{\sl}_3$, the polynomial part of $R_{J_{{\alpha_1},0}}^{{\alpha_1},7}(X)D_{{\alpha_1},7}$ is
\begin{align}
    R_{J_{{\alpha_1},0}}^{{\alpha_1},7}(X) = &X^{{-\alpha_2},2}X^{{-\alpha_1},2}X^{{\alpha_1+\alpha_2},3}
                                + 4X^{{1},3}(X^{{1},2})^2\nn\\
                              &- 4X^{{2},2}X^{{1},3}X^{{1},2}+(X^{{2},2})^2X^{{1},3}\nn\\
                               & - 4(X^{{1},1})^2X^{{-\alpha_1},2}X^{{\alpha_1},3}
                                - (2X^{{1},1})^2X^{{-\alpha_1-\alpha_2},2}X^{{\alpha_1+\alpha_2},3}\nn\\
                                &- 8X^{{1},3}X^{{1},2}(X^{{1},1})^2+\dots
\end{align}
and one sees that each of these terms is a product of monomials with loop degree strictly less than $\floor{7/2}$; this is part of a pattern, and the gap, $n-n/2 = n/2$ grows unboundedly with $n$. We define the map
\begin{equation}
  \label{eq:8}
  \begin{split}
  \nu:\g &\longrightarrow \wt{\Der}\Oa(\n_+)\\
  J_{a,n}&\longmapsto \ff{a}{b}{c}\sum_{k\geq \max(0,n)}X^{b,k-n}D_{c,k},
  \end{split}
\end{equation}
where $\ff{a}{b}{c}$ are the structure constants of $\og$.

As a side comment, note that from its very definition, if $a\in\upcirc{\Delta}_+$, $X^{a,0}$ is part of the coordinate system on $U$, while if $a\in\upcirc{\Delta}_-$ or $a\in\{1,\dots,\rank\g\}$, it is not. In order to take this fact into account, one should really consider as the lower bound of the sum in \cref{eq:8} the expression $\max(\eta(b),n+\eta(a))$, where $\eta:(\upcirc{\Delta}\cup\{1,\dots,\rank\g\})\to \{0,1\}$, defined as $\eta(\upcirc{\Delta}_+) = 0$, $\eta(\upcirc{\Delta}_-)=1$, $\eta(\{1,\dots,\rank\g\})=1$. In order to keep the notation cleaner, we will implicitly assume this below.

The following theorem makes the observation in \cref{eq:8} more precise, showing that apart from the leading monomial in each $P_{J_{a,n}}^{a,n}(X)$, the remaining part always has widening gap:
\begin{theorem*}[\cite{analogFFR}]\label{th1}
  For all $(a,n)\in\mathcal I\times \ZZ$,
  \begin{equation}
    \label{eq:thWG}
    \varrho(J_{a,n})-\nu(J_{a,n})=: \sum_{(b,m)\in\sfA}R^{b,m}_{J_{a,n}}(X)D_{b,m}\in\bDer\Oa(\n_+),
  \end{equation}
  where $R\in\Oa(\n_+)$.
\end{theorem*}

It is important to stress that the widening gap part may contain \emph{finite} sums and in particular, because of definition \cref{eq:8}, finite sums of quadratic terms of the form $X^{a,n}D_{a,n}$.

\subsection{Vertex algebras and splitting map}\label{VAsplit}

Following the finite-type construction, the next natural step would be to consider the corresponding vertex algebras and repeat the same construction as above.
As before, one can define the vacuum Verma module over the central extension by $K$ of the loop algebra of $\g$, $\g\otimes\CC[s,s\inv]$, which naturally carries the structure of a vertex algebra. We denote it by $\VV_0^{\g,K}$, where $\VV_0^{\g,K}[1]\simeq \g$.

Similarly to the finite-type construction, one also defines the Fock module for the $\bet\gam$-system on $\n_+$, $\sfM(\n_+)$, which has the structure of a vertex algebra.
Having in mind the idea of embedding expressions like the one in \cref{E11}, we need to enlarge this space by allowing infinite sums, hence we have to work in a certain completion of this space, $\tM(\n_+)$.

As before, from \eqref{eq:liehom1} one can introduce a map of vector spaces
\begin{align}
  \vartheta: \VV_0^{\g,K} \rightarrow \tM(\n_+),
\end{align}
from the vacuum Verma module over the double loop algebra of $\og$ at level $K$.
As in the finite-type case it turns out that non-negative products are not preserved.
Nevertheless, the remarkable result from \cite{analogFFR}, is that also in the affine setting there exists an analogue of the splitting map \eqref{eq:phimap}, namely
\begin{equation}
  \label{eq:phimap2}
  \varphi: \g \to \Omega_{\Oa(\n_+)},
\end{equation}
where $\Omega_{\Oa(\n_+)}$ is the space of 1-forms. It maps
\begin{align}\label{phiimage}
  J_{a,n}\mapsto \sum_{(b,m)\in\sfA}Q_{J_{a,n};b,m}(X)dX^{b,m}.
\end{align}
It has the property that the map
\begin{align}
  \label{eq:fullmap}
   \varrho +\varphi: \g \rightarrow \wt{\Der}\Oa(\n_+)\oplus\Omega_{\Oa(\n_+)},
\end{align}
can be lifted to a linear map
\begin{align}
  \label{eq:mapVA}
  \vartheta: \VV_0^{\g,0} \rightarrow \tM(\n_+),
\end{align}
from the vacuum Verma module over the \emph{loop} algebra of $\g$, \ie at zero level, so that the zeroth products of generators are preserved (cfr. \cite[Theorem 33]{analogFFR}):
\begin{align}
    \vartheta(J_{a,n}[-1]\vac)_{(0)} \vartheta(J_{b,m}[-1]\vac) = \vartheta(J_{a,n}[-1]\vac_{(0)}J_{b,m}[-1]\vac).
\end{align}

Unfortunately, the space $\tM(\n_+)$ is not a vertex algebra. For example, the would-be first products may result in ill-defined divergent sums. Consider the state $a=\sum_{a\in\Ia,k\geq 0} \gam^{a,k}[0]\bet_{a,k}[-1]\vac$; the first product of this element with itself would be
\begin{align}
  a_{(1)}a &= \sum_{a\in\Ia , k\geq 0} \gam^{a,k}[0]\bet_{a,k}[-1]\vac _{(1)} \sum_{b\in\Ia,j\geq 0} \gam^{b,j}[0]\bet_{b,j}[-1]\vac\nn\\
           &= \sum_{a\in\Ia,k\geq 0} \sum_{b\in\Ia,j\geq 0}\wick{\c1\gam^{a,k}[1]\c2\bet_{a,k}[0] \c2\gam^{b,j}[0]\c1\bet_{b,j}[-1]\vac}=\sum_{a\in\Ia}\sum_{k\geq 0}1\vac,
\end{align}
which clearly diverges. The quadratic infinite sums we saw above also suffer from this problem.

\subsection{Main result}
\label{sec:results}
On encountering such divergent sums, it is a standard trick in physics to attempt to cure them via $\zeta$-function regularisation. The main result of this paper is to prove the following observation first made in \cite{analogFFR}. Let $\reg$ denote the $\zeta$-function regularisation procedure, which we define in a moment. Then, we have the following
\begin{theorem*}\label{prop2intro}
  For all $J_{a,n}[-1]\vac,J_{b,m}[-1]\vac\in\VV_0^{\g,0}$ one has
  \begin{equation}
    \label{eq:0pr2bis}
    \textup{\texttt{reg}}[\vartheta(J_{a,n}[-1]\vac)_{(1)}\vartheta(J_{b,m}[-1]\vac)] = 0.
  \end{equation}
\end{theorem*}
In other words, after zeta-function regularisation, the divergent first products are not merely finite but, strikingly, are actually \emph{all} zero on the nose.
One should caveat that this statement does not immediately lead to a homomorphism of vertex algebras, essentially because the map $\reg$ does not respect associativity. However, the seemingly rather intricate cancellations involved are certainly highly suggestive.

In more detail, the $\zeta$-function regularisation procedure we employ works as follows.
We define a new vertex Lie algebra over the ring of rational function $\CC(z)$, whose products depend on a formal parameter $z$.
This parameter $z$ plays the role of a \emph{regulator}: this will enable us to identify the ill-defined quantities and ultimately to ``renormalise'' them. The map
\begin{align}
  \reg : \CC(z) \to \CC,
\end{align}
is defined through a series of steps: one first maps $z\mapsto e^{y} $ and then expands for small $y$. Then, after some manipulations of the resulting series, one extracts the constant term.
To see how it works, consider the following example
\begin{equation}
  \label{eq:exreg2}
  \begin{split}
    & \frac{z^2}{1-z^{2}} \leadsto \frac{e^{2y}}{1-e^{2y}}\\
    &\qquad \leadsto\frac{(1+2y+2y^2+\dots)}{1-1-2y-2y^2-\dots} = - \frac{1}{2y}\frac{(1+2y+\dots)}{(1+y+\dots)}\\
    &\qquad \qquad \leadsto -\frac{1}{2 y} (1+2y+\dots)(1-y+\dots)   = -\frac{1}{2y} -\frac12 + \Oa(y) \leadsto -\frac12
  \end{split}
\end{equation}
This is equivalent to removing the singular term and taking the standard limit $y\to0$ ($z\to 1$).

\medskip\smallskip
\begin{center}
  ***
\end{center}
The article is structured as follows.

In \cref{sec:def} we recall the definitions of vertex algebras and vertex Lie algebras and their relation.
In \cref{sec:Malgb}, we introduce the space $\hM_z[\leq 1]$, depending on the regulator $z$. We show that this space has the structure of a vertex Lie algebra over $\CC(z)$.
Finally, in \cref{sec:zeta} we introduce the regularisation procedure and use it to regularise the products of $\hM_z$. Section \ref{appB} is dedicated to the proof of the main theorem.

\section{Vertex (Lie) algebras}\label{sec:def}

\subsection{Definition and properties}
\label{sec:defpropVA}

\begin{definition}\label{def:VA}
  Given a graded vector space $\Vs=\sum_{k\in\ZZ}\Vs^{(k)}$ over a field $\KK$, a vector $\vac\in\Vs$ called the \emph{vacuum}, a map $T\in\End(\Vs)$ called \emph{translation}, and a map called the \emph{state-field correspondence} $Y:\Vs \to \End\Vs[[x,x\inv]]$, $a\mapsto Y(a,x):=\sum_{k\in\ZZ}a_{(k)}x^{-k-1}$, with $a_{(k)}\in\End\Vs$ of degree $\deg(a)-k-1$, \ie $a_{(k)}\Vs^{(n)}\subset \Vs^{(n+\deg(a)-k-1)}$, a \emph{vertex algebra} is the quadruple $(\Vs,\vac,T,Y(\bullet,x))$ satisfying the following axioms
  \begin{enumerate}[label=\roman*),ref=\it{\roman*)}]
  \item \label{VAi} Vacuum axiom: $Y(\vac,x)=\id_{\Vs}$;
  \item \label{VAii} Creation axiom:  for all $a\in\Vs$, $Y(a,x)\vac \in \Vs[[x]]$ and $Y(a,x)\vac\Big|_{x=0}=a$;
  \item \label{VAiii} Translation axiom: $T\vac=0$ and for all $a\in\Vs$, $[T,Y(a,x)] = \partial_xY(a,x)$;
  \item \label{VAiv} Borcherds' identities: for all $a,b\in\Vs$,
    \begin{align}
      \label{eq:borVA}
      \res_{x-y}\iota_{y,x-y}&f(x,y)Y(Y(a,x-y)b,y) =\nn \\
                             & \res_{x}\iota_{x,y}f(x,y)Y(a,x)Y(b,y) - \res_{x}\iota_{y,x}f(x,y)Y(b,y)Y(a,x),
    \end{align}
    where $f(x,y)$ is a rational function with poles at most at $x=0$, $y=0$ or $x-y=0$ and $\iota_{x,y}$ denote the formal power series expansion in the domain $\abs{x}>\abs{y}$.
  \end{enumerate}
\end{definition}

Given a state $a\in\Vs$, the field $Y(a,x)$ can be interpreted as the generating function of an infinite number of products, called the \emph{$n$-th products}, $\Vs\times\Vs\to \Vs$, $(a,b)\mapsto a_{(n)}b$ for all $n\in\ZZ$, such that $a_{(n)}b=0$ for sufficiently large $n$, where $a_{(n)}\in\End\Vs$ is called the \emph{$n$-th mode of $a$}.
It is possible to rephrase the axioms in terms of modes as follows
\textit{
  \begin{enumerate}[label=\roman*'),ref=\it{\roman*'}]
  \item \label{VAip} Vacuum axiom: for all $n\in\ZZ$, $\vac_{(n)}a=\delta_{n,-1}a$;
  \item \label{VAiip} Creation axiom: for all $a\in\Vs$, $n\in\ZZ_{\geq 0}$, $a_{(n)}\vac=0$ and  $a_{(-1)}\vac=a$;
  \item \label{VAiiip} Translation axiom: for all $a\in\Vs$, $n\in\ZZ$, $[T,a_{(n)}] = -na_{(n-1)}$;
  \item \label{VAivp} Borcherds' identities: for all $a,b\in\Vs$, $n,m\in\ZZ$,
    \begin{align}
      \label{eq:borVAp}
        \sum_{j=0}^{\infty}\binom{m}{j}(a_{(n+j)}b)_{(m+k-j)} =  \sum_{j=0}^{\infty}\binom{n}{j} &\Big( (-1)^j a_{(m+n-j)}(b_{(k+j)})\nn\\
                                                                                                 &\qquad  - (-1)^{j+n} b_{(n+k-j)}(a_{(m+j)})\Big).
    \end{align}
    where for any $m\in\ZZ$, $\binom{m}{k}=\frac{1}{k!}m(m-1)\dots(m-k+1)$ for all $k>0$ and $\binom{m}{0}=1$.
  \end{enumerate}
  }

  Setting $n=0$ in \cref{eq:borVAp}, we find the so-called \emph{commutator formula} for the modes,
  \begin{align}\label{commfor}
    [a_{(m)},b_{(k)}] = \sum_{j=0}^{\infty}(a_{(j)}b)_{(m+k-j)}.
  \end{align}
Expressing this in terms of state-field map, one finds the \emph{locality formula} (see \cite[Theorem 2.3]{KacVA}),
  \begin{align}\label{locform}
    (x-y)^N[Y(a,x),Y(b,y)] = 0, \qquad \text{for $N$ big enough}.
  \end{align}

Similarly, setting $m=0$ one obtains the \emph{associativity formula}, which can be seen as a recursive formula for the composition of modes,
 \begin{align}\label{eq:compmodes}
    (a_{(n)}b)_{(k)} = \sum_{j=0}^{\infty}\binom{n}{j} &\Big( (-1)^j a_{(n-j)}(b_{(k+j)})- (-1)^{j+n} b_{(n+k-j)}(a_{(j)})\Big).
  \end{align}
  In particular, setting $n=-1$ one gets
  \begin{equation}
    \label{eq:NO}
    Y(a_{(-1)}b,y) = Y(a,y)_+Y(b,y) + Y(b,y)Y(a,y)_-,
  \end{equation}
  where
  \begin{align}
    Y(a,x)_+:=\sum_{k < 0}a_{(k)}x^{-k-1}, \qquad Y(a,x)_-:=\sum_{k \geq 0}a_{(k)}x^{-k-1}.
  \end{align}
This can be generalised to define the so-called \emph{normal ordering of fields},
\begin{equation}
  \label{eq:no}
  \no{Y(a,x)Y(b,y)}:=\sum_{n\in\ZZ}\left( \sum_{m<0}a_{{m}}b_{{n}}x^{-m-1}+\sum_{m\geq 0}b_{(n)}a_{(m)}x^{-m-1}\right)y^{-n-1}.
\end{equation}
where the effect of the ordering is to move positive modes to the right, which in a field theory context can be interpreted as ``annihilation operators act first''.
This operation is right-associative, \ie $\no{Y(a,x)Y(b,y)Y(c,z)}=\no{Y(a,x)\no{Y(b,y)Y(b,z)}}$.

From these axioms, it is possible to show that $Y(a,x)b=e^{x T}Y(b,-x)a$ for all $a,b\in\Vs$. In the language of modes, this gives rise to the so-called \emph{skew-symmetry formula} for the products,
\begin{equation}
  \label{eq:skew}
  a_{(n)}b = -\sum_{j=0}^{\infty}\frac{(-1)^{n+j}}{k!}T^{j}(b_{(n+j)}a).
\end{equation}
For example, it turns out that the zeroth product is skew-symmetric modulo translations, while the first product is symmetric modulo translations.

\medskip

We now proceed to define vertex Lie algebras.
\begin{definition}\label{def:VLA}
  Given a graded vector space $\Ls=\sum_{k\in\ZZ}\Ls^{(k)}$ over a field $\KK$, a vector $\vac\in\Ls$ called the \emph{vacuum}, a map $T\in\End(\Ls)$ called \emph{translation}, and a map $Y_-:\Ls \to \End(\Ls)t\inv\CC[[t\inv]]$, $a\mapsto Y(a,x):=\sum_{k\geq 0}a_{(k)}x^{-k-1}$, with $a_{(k)}\in\End(\Ls)$ of degree $\deg(a)-k-1$, \ie $a_{(k)}\Ls^{(n)}\subset \Ls^{(n+\deg(a)-k-1)}$, a \emph{vertex Lie algebra} is the quadruple $(\Ls,\vac,T,Y_-(\bullet,x))$ satisfying the following axioms
  \begin{enumerate}[label=\roman*),ref=\it{\roman*}]
  \item \label{VLAi} $a_{(n)}\vac = 0$ for $n$ large enough;
  \item \label{VLAii} Translation axiom: $(T a)_{(n)}b = -n a_{(n-1)} b$;
  \item \label{VLAiii} Skew-symmetry axioms: $a_{(n)}b = - \sum_{k\geq 0}\frac{(-1)^{n+k}}{k!}T^k(b_{(n+k)}a)$;
  \item \label{VLAiv} Borcherds' identities: for all $a,b\in\Ls$,
      \begin{equation}
        \label{eq:borVLAp}
          \sum_{j=0}^{\infty}\binom{m}{j}(a_{(j)}b)_{(m+n-j)}c =  a_{(m)}b_{(n)}c - b_{(n)}a_{(m)}c
      \end{equation}
  \end{enumerate}
\end{definition}

It is clear that the polar part of a vertex algebra, obtained by forgetting about all terms with positive powers $x^n$, $n\geq 0$, in the state-field map $Y(\bullet,x)$, gives rise to a vertex Lie algebra.

To any vertex Lie algebra, one can associate a Lie algebra which is the Lie algebra of its modes, namely
\begin{align}
  \Lie(\Ls) = \Ls\otimes \CC((t))/\Im(T\otimes 1 + \id \otimes \partial_t).
\end{align}
with commutation relations given by
\begin{align}
  [A_{[m]},B_{[n]}] = \sum_{k\geq 0} \binom{m}{k}(A_{(n)}B)_{[m+n-k]},
\end{align}
where by $A_{[n]}$ we identify the image of $A\in\Ls$ in $\Lie(\Ls)$.
It is also possible to define the left-adjoint functor to the one sending a vertex algebra to its polar part (for more detail cfr. \cite{frenkelvertex}). The resulting space is the \emph{universal enveloping vertex algebra}
\begin{align}
  \VV(\Ls) : = U(\Lie(\Ls))\otimes_{U(\Lie(\Ls))_+}\CC,
\end{align}
which has the structure of a vertex algebra.

\section{The vertex Lie algebra $\hM_z[\leq1]$}\label{sec:Malgb}
\subsection{Heisenberg algebra and Fock module}\label{sec:Heis}
Let $\g$ be an affine Kac-Moody algebra with triangular decomposition $\g\simeq \n_-\oplus\h\oplus\n_+$, whose underlying finite algebra is denoted by $\og$.
Let $\sfA = \{(a,0)\}_{a\in\upcirc\Delta_+}\cup (\Ia\times\ZZ_{\geq 1})$ be an index set indexing a basis of $\n_+$, where $\upcirc{\Delta}_+$ is the set of positive roots of $\og$ and $\Ia = (\upcirc{\Delta}\setminus\{0\})\cup\{1,\dots,\rank\og\}$.

Let $z$ be a formal parameter. We denote by $\CC[z]$, the polynomial ring in $z$ with complex coefficients. Given $M,N \in \ZZ$ and $(a,m) \in \sfA$, consider the free unital associative algebra $\sfH_z$ generated by $\bet_{a,m}[M]$, $\gam^{a,m}[N]$ and $\bsf{1}$ quotiented by the ideal generated by the commutation relations
\begin{equation}
  \label{eq:commrelH}
  \begin{gathered}
    [\bet_{a,m}[M],\bet_{b,n}[N]]=0, \qquad  \qquad [\gam^{a,m}[M],\gam^{b,n}[N]]=0,\\
    [\bet_{a,m}[M],\gam^{b,n}[N]]= z^m\delta_{N+M,0}\delta^{b}_{a}\delta^n_{m}\bsf{1}.
  \end{gathered}
\end{equation}
This algebra can be seen as a ``deformation'' of the \emph{Heisenberg algebra} $\sfH$, that can be recovered by taking the limit $z\to 1$. We introduce the parameter $z$ with the role of \emph{regulator}: its meaning will be clear in the following sections.

In this section, unless otherwise stated, we work over the ring $\CC[[z]]$ of formal power series. As we will see, in certain special cases we can work over $\CC[z]$ or $\CC(z)$.

This algebra represents a system of free fields as it can be decomposed into $\sfH_z\cong \sfH_z^+\otimes\sfH_z^-$, where
\begin{align}
  \label{eq:Hsubalg}
  \sfH_z^-\simeq_{\CC}\{ \gam^{a,m}[N],\bet_{a,m}[N-1]\}_{(a,m)\in \sfA;N\leq 0}, \\
  \sfH_z^+\simeq_{\CC}\{ \gam^{a,m}[N],\bet_{a,m}[N-1]\}_{(a,m)\in \sfA;N > 0},
\end{align}
are called the \emph{creation} and \emph{annihilation} subalgebras, respectively.

Introducing a vacuum vector $\vac$, we call $\sfM_z$ the induced $\sfH_z$-module annihilated by $\sfH_z^+$, on which $\bsf{1}\vac=\vac$.

Denoting by $Q$ the root lattice of $\g$, there is the $\ZZ\times Q$ gradation of $\sfH_z$ and consequently of $\sfM_z$, in which $\bet_{a,n}[N]$ has grade $(N,\alpha)$ and $\gam^{a,n}[N]$ has grade $(N,-\alpha)$, whenever $J_{a,n}\in\g_{\alpha}$ and $\vac$ has degree $(0,0)$.
The natural depth gradation of $\sfM_z$ is given by
\begin{equation}
  \label{eq:depthgr}
  \sfM_z=\bigoplus_{i=0}^{\infty}\sfM_z[i]
\end{equation}
and a corresponding filtration $\sfM_z[\leq\! K]=\bigoplus_{i=0}^K\sfM_z[i]$ for $K\geq 0$.

More explicitly, the space $\sfM_z[0]$ is spanned by elements of the form $R(\gam[0])\vac$, while $\sfM_z[1]$ by finite linear combinations of elements of the form
\begin{equation}
  \label{eq:basM1}
  P(\gam[0])\bet_{a,m}[-1]\vac + Q(\gam[0])\gam^{a,m}[-1]\vac
\end{equation}
where $P,Q,R$ are polynomials in $\CC[\gam^{a,n}[0]]_{(a,n)\in \sfA}$.

\subsection{Vertex algebra structure on $\sfM_z$}\label{sec:VAstr}
The space $\sfM_z$ is called the \emph{Fock module of the $\bet\gam$-system of free fields} and it has the structure of a vertex algebra over $\CC[z]$, as defined in \cref{sec:defpropVA}. The state-field-map $Y$ is
\begin{equation}\label{eq:Ymap}
  Y: \sfM_z \to \End\sfM_z((x)), \qquad  A \mapsto Y(A,x):=\sum_{k\in\ZZ}A[k] x^{-k-1}
\end{equation}
satisfying the axioms \labelcref{VAi,VAii,VAiii,VAiv}, where the $k$th-mode is the map in $\End(\sfM_z)$ denoted by
\begin{equation}
  \label{eq:nmode}
  \sfM_z \to \End \sfM_z, \qquad a \mapsto a[k], \qquad  k\in\ZZ.
\end{equation}
The fields are
\begin{equation}
  \label{eq:fields}
  \bet_{a,n}(x):=\sum_{k\in\ZZ}\bet_{a,n}[k]x^{-k-1}, \qquad \gam^{a,n}(x):=\sum_{k\in\ZZ}\gam^{a,n}[k]x^{-k},
\end{equation}
where $\bet_{a,n}[k]:=(\bet_{a,n}[-1]\vac)_{(k)}$ and $\gam^{a,n}[k]:=(\gam^{a,n}[0]\vac)_{(k-1)}$.
Composite fields are obtained by using iteratively \cref{eq:compmodes}.
The translation map is defined as follows
\begin{equation}
  \label{eq:defT}
  T\gam^{a,n}[N]\vac = -(N-1)\gam^{a,n}[N-1]\vac, \qquad T\bet_{a,n}[N]\vac= -N \bet_{a,n}[N-1]\vac.
\end{equation}
for $N\in\ZZ_{\leq 0}$ and $T\vac=0$.

\subsection{Completion of the Fock module}\label{completion}
Having in mind to lift the homomorphism in \cref{eq:liehom1} to the vertex algebra case, we also need to be able to work with certain infinite sums.

This is achieved by working in a suitable completion of the space $\sfM_z$. This means that infinite sums will be allowed provided they truncate to finite ones, modulo terms containing $\bet_{a,m}$ and $z^m$, for $m$ big enough.

The explicit construction goes as follows. Let us denote by $\sfH^-_{z,\geq k}$ the two-sided ideal in $\sfH_z^-$ generated by $\{\bet_{a,m}[N]: m \geq k,  N\in\ZZ\}$.
Let
\begin{equation}
  \Is_z[\leq M]_k:=\sfM_z[\leq M] \cap (\sfH^-_{z,\geq k}\vac),
\end{equation}
for some $M\in\ZZ_{\geq 0}$. Therefore, $\Is_z[\leq M]_k$ is the subspace of $\sfM_z[\leq M]$ spanned by monomials of depth less or equal to $M$ in the creation operators with some factor $\bet_{a,m}[N]$, with $m \geq k$, $N\in\ZZ$.
One has
\begin{equation}
  \Is_z[\leq M]_0\supset \Is_z[\leq M]_1\supset \Is_z[\leq M]_2 \supset \dots
\end{equation}
with $\bigcap_{i=0}^{\infty} \Is_z[\leq M]_i=\{0\}$. One defines the completed subspaces
\begin{equation}\label{invlim}
  \tM_z[\leq M] := \varprojlim_{k} \sfM_z[\leq M] / \Is_z[\leq M]_{k}.
\end{equation}

To give an element of this inverse limit means to give an element of each space of the inverse system, in a manner compatible with the projection maps between them. In this sense one allows infinite sums: every element of the sum is well-defined because each of its truncations is well-defined.
Finally, one can consider a completion in the depth direction, from the system of inclusions
\begin{align}
  \tM_z[0] \subset \tM_z[\leq 1]\subset \dots
\end{align}
by taking the direct limit
\begin{align}\label{dirlim}
  \tM_z := \varinjlim_{M} \tM_z[\leq M].
\end{align}
Any element of $\tM_z$ is therefore a well-defined element in $\tM_z[\leq M]$ for some $M$.
We have the following result:
\begin{proposition}\label{propzz}
    The space $(\tM_z[\leq 1],\vac,T,Y(\bullet,x))$ is a vertex Lie algebra.
\end{proposition}
\begin{proof}
  The results follows from the fact that each space $\sfM_z[\leq 1] / \Is_z[\leq 1]_{k}$ has the structure of a vertex Lie algebra over $\CC[z]/z^k\CC[z]$, with $k\in\ZZ_{\geq 0}$. Therefore, the inverse limit \cref{invlim} defines a vertex Lie algebra over the inverse limit of rings $\CC[z]/z^k\CC[z]$, which is the ring of power series $\CC[[z]]$.
\end{proof}

For example, for $p(z),q(z)\in\CC[[z]]$, one would get
  \begin{align}
    \label{eq:exampleof1pr}
       (p(z)\sum_{k\geq 1} \gam^{a,k-1}[0]&\bet_{b,k}[-1]\vac)_{(1)}q(z)\sum_{j\geq 0}\gam^{a,j+1}[0]\bet_{b,j}[-1]\vac\nn\\
      & =p(z)q(z)\sum_{k\geq 1} \sum_{j\geq 1}\wick{\c1\gam^{a,k-1}[1]\c2\bet_{b,k}[0]\c2\gam^{c,j+1}[0]\c1\bet_{d,j}[-1]\vac}\nn\\
      & = -\delta^a_d\delta^c_b p(z)q(z)\sum_{k\geq 1}z^{2k - 1}\vac,
  \end{align}
  whose $k$-truncations are all well defined over $\CC[z]/z^k\CC[z]$, $k\in\ZZ_{\geq 0}$.

From this simple example we see the fact mentioned above: when $z\to 1$, the vertex algebra structure breaks down since ill-defined quantities start to appear.

\subsection{States with widening gap}\label{sec:WG}
We now restrict the space $\tM_z$, by excluding all those sums that give rise to infinite power series in $z$, in such a way that the vertex algebra structure is preserved.
This follows by extending the idea of widening gap introduced in \cite{analogFFR} to this setting.

A family of polynomials $\{P^{a,n}\}_{(a,n)\in\sfA}$ in $\CC[\gam^{b,m}]_{(b,m)\in\sfA}$ has widening gap if for all $K\geq 1$, there exists a $\overline{n}\in\ZZ_{>0}$, such that for all $n\geq \overline{n}$
\begin{equation}
  \label{eq:WG}
  P^{a,n}(\gam) \in\CC[\gam^{b,m} : m < n - K, b\in\Ia].
\end{equation}

Define $\bM_z\subset\tM_z$, as the space spanned by sums of the form
\begin{equation}
  \label{eq:genelM}
  \sum_{(a_i,m_i)_{i=1,\dots,n}\in\sfA}P^{a_1,m_1}(\gam) \cdots P^{a_n,m_n}(\gam) \bet_{a_1,m_1}[-N_1] \cdots \bet_{a_n,m_n}[-N_n] \vac
\end{equation}
where $N_i\in\ZZ_{>0}$, $i=1,\dots,n$ and the polynomials $P^{a_i,m_i}(\gam)$ have widening gap.
By construction, we have that $\sfM_z\subset \bMz$, since any finite sum has obviously a widening gap.

We have the following useful result,
  \begin{lemma}\label{lem:der}
    Given a collection of polynomials with widening gap $\{P^{b,m}(\gam)\}_{(b,m)\in\sfA}$,
    \begin{equation}
      \label{eq:der}
      \bet_{a,n}[N]P^{b,m} (\gam)\vac= z^{n}\frac{\partial P^{b,m}(\gam)}{\partial \gam^{a,n}[-N]}\vac
    \end{equation}
    is again a collection of polynomials with widening gap, with $(a,n)\in\sfA$ and $N\in\ZZ_{\geq 0}$.
  \end{lemma}

The following statement characterises the space of states with widening gap
\begin{lemma}\label{lem:MVA}
  The space $(\bMn1,\vac,T,Y(\bullet,z))$ is a vertex Lie algebra.
\end{lemma}

The proof is essentially the same of \cite[Lemma 21]{analogFFR}, with the addition of the regulator which is this setting does not produce any particular difference. As a matter of fact, when taking the limit $z\to 1$, the vertex Lie algebra structure is not spoiled.

\subsection{The space $\hMn{1}$}\label{sec:M}
Recall from \cref{sec:WGsub} that the image of the element $J_{a,n}\in\g$ under the map $\varrho$ from \cref{eq:liehom1} contains sums that don't develop a widening gap, namely the image of the map $\nu$ from \cref{eq:8}.
To take care of these terms, in this section we introduce a slightly bigger space obtained by adjoining the specific type of infinite sums from \cref{eq:8}.

Define by $\Oa_z=\CC(z)[\gam^{a,n}[0]]_{(a,n)\in\sfA}$ the ring of polynomials in the generators $\gam^{a,n}[0]$, $(a,n)\in\sfA$ with coefficients in $\CC(z)$, the ring of rational functions.

We define the space of \emph{quadratic infinite sums} $\sfQ_z$ as follows
\begin{equation}
  \label{eq:Qdef}
  \sfQ_z := \left\{ \sum_{k\geq \max(0,n)} z^{\alpha k}\gam^{a,k-n}[0]\bet_{b,k}[-1]\vac : \alpha\in\ZZ_{\geq 0}, n\in\ZZ \right\},
\end{equation}
which clearly are sums that do not develop a widening gap.

Working now over $\CC(z)$, we introduce the direct sum
\begin{equation}
  \label{eq:Mhat}
  \hMn{1} := \sfQ_z\oplus \bMn{1} \subset \tMn{1}.
\end{equation}
Elements in $\hM_z[0]$ are finite sums of terms $R(\gam)\vac$, $R(\gam)\in\Oa_z$, while elements in $\hM_z[1]$ are of the form
\begin{equation}
  \begin{split}
    &\sum_{k\geq \max(0,n)} z^{\alpha k}\gam^{a,k-n}[0]\bet_{b,k}[-1]\vac \\
    &\qquad+ \sum_{(a,n)\in\sfA}P^{a,n}(\gam)\bet_{a,n}[-1]\vac + \sum_{b,m} Q_{b,m}(\gam)\gam^{b,m}[-1]\vac,
  \end{split}
\end{equation}
where the second possibly infinite sum is over a family of polynomials $P^{a,n}(\gam)\in\Oa_z$ with widening gap, while the last sum is finite and $Q(\gam)\in\Oa_z$.

\subsection{Vertex Lie algebra $\hMn{1}$}
\label{sec:VLA}

Extending the definition of the modes of the states from \cref{sec:VAstr} to this space, we see that some products can generate infinite power series in $z$.
In particular, we find that the first products between quadratic infinite sums have coefficients in $\CC[[z]]$:
    \begin{align}\label{eq:1Mhat}
         (&\sum_{k\geq \max(0,n)}z^{\alpha k}\gam^{a,k-n}[0]\bet_{b,k}[-1]\vac)_{(1)}\sum_{j\geq \max(0,m)}z^{\beta j}\gam^{c,j-m}[0]\bet_{d,j}[-1]\vac\nn\\
        =& \sum_{k\geq \max(0,n)}\sum_{j\geq \max(0,m)}z^{\alpha k}z^{\beta j}\wick{\c1\gam^{a,k-n}[1]\c2\bet_{b,k}[0]\c2\gam^{c,j-m}[0]\c1\bet_{d,j}[-1]\vac}\nn\\
      =& -\delta^a_d\delta^c_b\delta_{m+n,0}z^{-n(\beta +1)}\sum_{k\geq  \max(0,m+n)}z^{k(\alpha+\beta+2)}\vac \\
          &\qquad -\delta^a_d\delta^c_b\delta_{m+n,0}z^{-n(\beta +1)}\sum_{k\geq  \max(0,n,m+n)}^{\max(0,m+n)-1}z^{k(\alpha+\beta+2)}\vac.
    \end{align}
In the last step, note that the second sum is quadratic but finite, hence it is a well defined element in $\CC[z]\subset\CC(z)$.
Conversely, the first term is an infinite sum in $\CC[[z]]$. Crucially, we can regard it as the expansion of a rational function for $\abs{z}< 1$. By doing this, we can rewrite the first term as
\begin{equation}\label{resummation}
    \sum_{k\geq \max(0,m+nn)}z^{k(\alpha+\beta+2)} = \frac{z^{\max(0,m+n)(\alpha+\beta+2)}}{1-z^{\alpha+\beta+2} }\in\CC(z).
\end{equation}

As a result, we have the following
\begin{proposition}\label{VAAL}
  Regarding infinite sums as the small-$z$ expansions of rational functions, $(\hMn{1},\vac,T,Y(\bullet,x))$ is a vertex Lie algebra over $\CC(z)$.
\end{proposition}

\begin{proof}

We first show that the products on $\hM_z$ close over $\CC[[z]]$.

Since we are considering the first two graded subspaces $\hM_z[0]$ and $\hM_z[1]$, we restrict our analysis on the only two non-trivial products on this space: the ${}_{(0)}$ and ${}_{(1)}$ product. In more explicit terms, for $f\in\hM_z[0]$ and $u,v\in\hM_z[1]$, we have the following possible non-trivial combinations
\begin{align}\label{ntrivial}
    f_{(0)} u \in \hM_z[0] ,\qquad u_{(0)} f \in \hM_z[0] ,\qquad u_{(0)}v \in \hM_z[1], \qquad u_{(1)}v \in \hM_z[0].
\end{align}

Moreover, we just need to focus on cross products, \ie products between elements in $\sfQ_z$ and $\bMn{1}$, and products of two infinite quadratic sums in $\sfQ_z$. The closure of the products between elements in $\bMn{1}$ follows from the fact that it already has the structure of a vertex Lie algebra, as pointed out in \cref{lem:MVA}.

Let us start considering the ${}_{(0)}$ products.
It is easy to show that expressions like $\npr{0}{\sfQ_z}{R(\gam)}\vac$ and $\npr{0}{\sfQ_z}{R(\gam)}\gam[-1]\vac$ or those with the factors flipped close in $\hMn{1}$ for any element $R(\gam)\in\Oa_z$, as they give rise to finite sums of polynomials of depth 0 or 1 which are well-defined elements in $\hMn{1}$, respectively.

    We need to show the closure for $\npr{0}{\sfQ_z}{\sfQ_z}$ and $\npr{0}{\sfQ_z}{\bMn{1}}$. The former reads
    \begin{align}\label{eq:qq}
      & (\sum_{k\geq \max(0,n)}z^{\alpha k}\gam^{a,k-n}[0]\bet_{b,k}[-1]\vac)_{(0)}\sum_{j\geq\max(0,m)}z^{\beta j}\gam^{c,j-m}[0]\bet_{d,j}[-1]\vac \nonumber\\
      &=\sum_{k\geq \max(0,n)}\sum_{j\geq \max(0,m)}z^{\alpha k}z^{\beta j}\wick{\gam^{a,k-n}[0]\c1\bet_{b,k}[0]\c1\gam^{c,j-m}[0]\bet_{d,j}[-1]\vac}\nonumber\\
      & + \sum_{k\geq \max(0,n)}\sum_{j\geq \max(0,m)}z^{\alpha k}z^{\beta j}\wick{\c1\gam^{a,k-n}[1]\bet_{b,k}[-1]\gam^{c,j-m}[0]\c1\bet_{d,j}[-1]\vac}\nonumber\\
      = & +\delta^c_bz^{-m(\alpha+1)}\sum_{k\geq \max(0,m,m+n)}z^{(\alpha+\beta+1) k}\gam^{a,k-(m+n)}[0]\bet_{d,k}[-1]\vac\nonumber\\
      &\qquad -\delta^a_dz^{-n(\beta+1)}\sum_{k\geq \max(0,n,m+n)}z^{(\alpha+\beta+1) k}\gam^{c,k-(m+n)}[0]\bet_{b,k}[-1]\vac\nn\\
      = & \delta^c_b z^{-m(\alpha+1)}\sum_{k\geq \max(0,m+n)}z^{(\alpha+\beta+1) k}\gam^{a,k-(m+n)}[0]\bet_{d,k}[-1]\vac\nn \\
      & \qquad -\delta^a_d z^{-n(\beta+1)}\sum_{k\geq \max(0,m+n)}z^{(\alpha+\beta+1) k}\gam^{c,k-(m+n)}[0]\bet_{b,k}[-1]\vac\nn\\
      & + \delta^c_b z^{-m(\alpha+1)}\sum_{k\geq \max(0,m,m+n)}^{\max(0,m+n)-1}z^{(\alpha+\beta+1) k}\gam^{a,j-(m+n)}[0]\bet_{d,k}[-1]\vac\nn\\
      &\qquad - \delta^a_d z^{-n(\beta+1)}\sum_{k\geq \max(0,n,m+n)}^{\max(0,m+n)-1}z^{(\alpha+\beta+1) k}\gam^{c,j-(m+n)}[0]\bet_{b,k}[-1]\vac
    \end{align}

    In the last step, we have further decomposed the sums into two terms that manifestly live in $\sfQ_z$, while the rest are finite sums. Note \emph{en passant} the reason why in the definition \eqref{eq:Qdef} we had to include the power of $z^{\alpha}$; indeed, setting $\alpha,\beta=0$, we see that the first two terms in the last expression would not be well-defined.

    Consider now the product between one element in $\sfQ_z$ and one infinite sum with widening gap in $\bM_z[1]$,
    \begin{align}\label{eq:6}
      (&\sum_{k\geq \max(0,m)} z^{\alpha k}\gam^{a,k-m}[0]\bet_{b,k}[-1]\vac)_{(0)}\sum_{(c,n)\in\sfA}P^{c,n}\bet_{c,n}[-1]\vac\nn\\
       &=  \sum_{k\geq \max(0,m)} z^{\alpha k}\wick{\gam^{a,k-m}[0]\c1\bet_{b,k}[0]\sum_{(c,n)\in\sfA}\c1 P^{c,n}\bet_{c,n}[-1]\vac}\nn\\
       &\qquad  +\sum_{k\geq \max(0,m)} z^{\alpha k}\wick{\c1\gam^{a,k-m}[1]\bet_{b,k}[-1]\sum_{(c,n)\in\sfA}P^{c,n}\c1\bet_{c,n}[-1]\vac}\nn\\
       &= \sum_{k\geq \max(0,m)}\sum_{(c,n)\in\sfA} z^{k(\alpha +1)}\frac{\partial P^{c,n}}{\partial \gam^{b,k}}\gam^{a,k-m}[0]\bet_{c,n}[-1]\vac\nn\\
       & \qquad-\sum_{k\geq \max(0,m)}\sum_{(c,n)\in\sfA} z^{k(\alpha+1) -m}\delta_{k-m,n}\delta^a_c P^{c,n}\bet_{b,k}[-1]\vac
    \end{align}
    The first sum is well defined because, from \cref{lem:der}, it has widening gap. The second sum is non-zero only if $n>-m$, for all $n,m\in\ZZ$: this ensures that the combination $k=n+m$ is non-negative, and therefore also the second sum develops widening gap.

   Let us consider now the $_{(1)}$ product.
    We need to show the closure only for products of elements of depth 1, namely $\npr{1}{\sfQ_z}{\sfQ_z}$ and $\npr{1}{\sfQ_z}{\bM_z[1]}$. The first case is
    \begin{align}\label{eq:7}
      & (\sum_{k\geq \max(0,n)}z^{\alpha k}\gam^{a,k-n}[0]\bet_{b,k}[-1]\vac)_{(1)}\sum_{j\geq \max(0,m)}z^{\beta j}\gam^{c,j-m}[0]\bet_{d,j}[-1]\vac\nn\\
      =& \sum_{k\geq \max(0,n)}\sum_{j\geq \max(0,m)}z^{\alpha k}z^{\beta j}\wick{\c1\gam^{a,k-n}[1]\c2\bet_{b,k}[0]\c2\gam^{c,j-m}[0]\c1\bet_{d,j}[-1]\vac}\nn\\
      =& -\delta^a_d\delta^c_b\delta_{m+n,0} z^{-n(\beta+1)} \sum_{k\geq \max(0,n,m+n)}z^{k(\alpha+\beta+2)}\vac
    \end{align}

    Consider now the product between one element in $\sfQ_z$ and one infinite sum with widening gap in $\bM_z[1]$,
    \begin{equation}
      \label{eq:deriv}
      \begin{split}
        &(\sum_{k\geq \max(0,m)} z^{\alpha k}\gam^{a,k-m}[0]\bet_{b,k}[-1]\vac)_{(1)}\sum_{(c,n)\in\sfA}P^{c,n}\bet_{c,n}[-1]\vac\\
        =&\sum_{k\geq \max(0,m)} \sum_{(c,n)\in\sfA}z^{\alpha k}\wick{\c2\gam^{a,k-m}[1]\c1\bet_{b,k}[0]\sum_{(c,n)\in\sfA}\c1 P^{c,n}\c2\bet_{c,n}[-1]\vac}\\
        = &-\sum_{k\geq \max(0,m)} \sum_{(c,n)\in\sfA}z^{k(\alpha +2)-n}\delta_{k-m,n}\delta^a_c \frac{\partial P^{c,k-m}}{\partial \gam^{b,k}}\vac\\
      \end{split}
    \end{equation}
Since this sum is finite, it represents a well-defined element in $\hM_z[0]$. All other combinations of elements give rise to well-defined elements in $\hMn{1}$.

This shows that the products close on $\hMn{1}$ over $\CC[[z]]$. In particular, since $\tM_z[\leq 1]$ has the structure of a vertex Lie algebra (cfr. \cref{propzz}), this implies that $\hM_z[\leq 1]$ is a vertex Lie subalgebra of $\tM_z[\leq 1]$ over $\CC[[z]]$, as all axioms (\ref{VLAi})-(\ref{VLAiv}) are satisfied.

These products also close on $\hM_z[\leq 1]$ over $\CC(z)$, when regarding infinite sums as expansions of rational functions. Indeed, looking at the calculations above, the only difference is in the first product between quadratic states in \cref{eq:7}, which reads
\begin{align}\label{eq:7bis}
  &(\sum_{k\geq \max(0,n)}z^{\alpha k}\gam^{a,k-n}[0]\bet_{b,k}[-1]\vac)_{(1)}\sum_{j\geq \max(0,m)}z^{\beta j}\gam^{c,j-m}[0]\bet_{d,j}[-1]\vac\nn\\
  &= -\delta^a_d\delta^c_b\delta_{m+n,0}z^{-n(\beta+1)}\left(\frac{z^{\max(0,n)(\alpha+\beta+2)}}{1-z^{\alpha+\beta+2} } +\sum_{k\geq  \max(0,n,m+n)}^{\max(0,m+n)-1}z^{k(\alpha+\beta+2)}\right)\vac.
\end{align}

To finally prove the statement, one has to ensure that the vertex Lie algebra axioms are still satisfied after regarding infinite sums as expansions of rational functions. This is the case for vacuum, translation and skew-symmetry as they are all equality involving single products. More subtle is the case for Borcherds' identities, since nested products appear. However, as pointed out above, resummation is only needed when computing the first products of quadratic states.
    Writing down explicitly the identities \eqref{eq:borVLAp} for all possible combinations of the products \eqref{ntrivial}, we find that the only non-trivial identity which presents a nesting of first products is
\begin{equation}
    u_{(1)}(v_{(1)}w) - v_{(1)}(u_{(1)}w) = (u_{(1)}v)_{(1)}w,
\end{equation}
for $u,v,w\in\sfQ_z$. However, as the first product of such states is proportional to the vacuum and the action of positive modes on the vacuum is zero, this identity is trivially satisfied.
\end{proof}

\section{Regularising the products}\label{sec:zeta}

In \cref{completion,sec:WG}, we defined the completion $\tM_z$ of the Fock module of a Heisenberg algebra and we restricted it to the subspace $\bMn{1}$ of sums with widening gap which has the structure of a vertex Lie algebra.
Then, in the section above we have introduced the space of main interest for this work, the vertex Lie algebra $\hMn{1}$ over $\CC(z)$, which is spanned by a specific type of quadratic infinite sums and possibly infinite sums with widening gap. We will now proceed to regularise the products.

\subsection{$\zeta$-function regularisation}
For any given expression in $\CC(z)$, we introduce the following regularisation procedure
\begin{equation}\label{reg}
    \reg: \CC(z) \to \CC
\end{equation}
defined as follows
  \begin{enumerate}[label=\itshape\roman*)]
  \item perform the transformation $z\to e^{y}$\label{reg1};
  \item power expand the resulting term for small values of $y$\label{reg2};
  \item  regard the result as the ratio of Laurent series, which is again a Laurent series\label{reg3};
  \item remove the singular terms and perform the limit $y\to 0$, \ie $z\to 1$ to remove the regulator. This is equivalent to extracting the constant term of the series obtained\label{reg4}.
  \end{enumerate}

  As an example, consider the following
  \begin{align}
    \label{eq:exreg2bis}
      & \frac{z^{3}}{1-z^{2}} \leadsto \frac{e^{3y}}{1-e^{2y}}\\
      &\qquad \leadsto\frac{(1+3y+9y^2/2+\dots)}{1-1-2y-2y^2-\dots} = - \frac{1}{2y}\frac{(1+3y+\dots)}{(1+y+\dots)}\\
      &\qquad \qquad \leadsto -\frac{1}{2 y} (1+3y+\dots)(1-y+\dots)   = -\frac{1}{2y} -1 + \Oa(y) \leadsto -1
  \end{align}
  where each arrow corresponds to one of the steps above. Therefore we would write
  \begin{equation}
    \label{eq:2}
    \reg\left[\frac{z^3}{1-z^{2}} \right] =-1 .
  \end{equation}
In particular, for any polynomial or rational function which is regular at $z=1$, this procedure is equivalent to the evaluation $z= 1$. For this reason, roughly speaking, we can regard this procedure as a ``renormalised''  version of the limiting procedure $z\to 1$.

\subsection{Regularisation of the first products}

Recall from \cref{eq:fullmap} the definition of the Lie algebra map
\begin{align}\label{LAmap}
 \varrho + \varphi : \g \to \wt{\Der}\Oa(\n_+)\oplus\Omega_{\Oa(\n_+)}.
\end{align}

We have the embedding into the Fock space
\begin{align}\label{jota}
  \jota: \wt{\Der}\Oa(\n_+)\oplus\Omega_{\Oa(\n_+)}\hookrightarrow\hM_z[\leq 1]
\end{align}
by simply replacing $X^{a,n}$ with $\gam^{a,n}[0]$, $D_{a,n}$ with $\bet_{a,n}[-1]$ and $dX^{a,n}$ with $\gam^{a,n}[-1]$.

By identifying $\VV_0^{k,0}[1]\simeq \g$ and then composing the Lie algebra map \eqref{LAmap} with the embedding \cref{jota} we obtain a map
  \begin{equation}
    \label{eq:theta}
    \vartheta:= \jota \circ (\varrho + \varphi) :\VV_0^{\g,0}[\leq1] \longrightarrow \hM_z[\leq 1].
  \end{equation}
More explicitly, any element $J_{a,n}[-1]\vac\in\VV_0^{\g,0}[1]$ gets mapped to
  \begin{equation}
    \label{eq:mapsto}
    \begin{split}
        \ff{a}{b}{c}&\sum_{k\geq \max(0,n)}\gam^{b,k-n}[0]\bet_{c,k}[-1]\vac \\
                  &\qquad+ \sum_{(b,m)\in\sfA}R_{J_{a,n}}^{b,m}(\gam)\bet_{b,m}[-1]\vac + \sum_{(b,m)\in\sfA}Q_{J_{a,n};b,m}(\gam)\gam^{b,m}[-1]\vac.
    \end{split}
  \end{equation}
  where $R^{b,m}_{J_{a,n}}(\gam) = \jota(R^{b,m}_{J_{a,n}}(X))$ from \cref{eq:thWG} and $Q_{J_{a,n};b,m}$ is the image of $\jota\circ\varphi$, cfr. \cref{phiimage}, while $\vac\mapsto\vac$.

We have the following
  \begin{lemma}
  \label{prop1}
    For any two states $J_{a,n}[-1]\vac, J_{b,m}[-1]\vac\in\VV_0^{\g,0}[\leq 1]$ one has
    \begin{equation}
      \label{eq:0pr}
      \textup{\texttt{reg}}[\vartheta(J_{a,n}[-1]\vac)_{(0)}\vartheta(J_{b,m}[-1]\vac)] = \textup{\texttt{reg}}[\vartheta(J_{a,n}[-1]\vac_{(0)}J_{b,m}[-1]\vac)].
    \end{equation}
  \end{lemma}

    The proof is essentially the same of \cite[Theorem 33]{analogFFR}. In particular the fact that there are no possible double contractions implies that there will never be terms of the form $q(z)\vac$, with $q\in\CC(z)$. For this reason, the action of the regularisation procedure is simply to compute the limit $z\to 1$, \ie it is equivalent of working in the unregulated setting.

  We will now move our attention to first products. As illustrated above, in this case double contractions can appear and the regularisation procedure becomes of central importance.
The main result of the paper is the following
\begin{theorem}\label{prop2}
  For any two states $J_{a,n}[-1]\vac, J_{b,m}[-1]\vac\in\VV_0^{\g,0}[\leq 1]$ one has
  \begin{equation}
    \label{eq:0pr2}
    \textup{\texttt{reg}}[\vartheta(J_{a,n}[-1]\vac)_{(1)}\vartheta(J_{b,m}[-1]\vac)] = 0.
  \end{equation}
\end{theorem}

\begin{proof}
  The proof occupies \cref{appB}.
\end{proof}

As an example, consider the elements from $\VV_0^{\hsl_2,0}$
\begin{align}
    J_{E,2}[-1]\vac &\xmapsto{\vartheta} \bet_{E,2}[-1]\vac -\sum_{k\geq 5} \gam^{F,k-2}[0]\bet_{H,k}[-1] \vac +2 \sum_{k\geq 5}\gam^{H,k-2}[0]\bet_{E,k}[-1]\vac +\dots\nn\\
    J_{F,-2}[-1]\vac &\xmapsto{\vartheta} -14\gam^{E,2}[-1]\vac -\sum_{k\geq 1} \gam^{E,k+2}[0]\bet_{H,k}[-1] \vac +2 \sum_{k\geq 1}\gam^{H,k+2}[0]\bet_{F,k}[-1]\vac +\dots
\end{align}
where we only wrote the quadratic infinite sums and the other terms that could contribute with terms of the form $\CC(z)\vac$ in the computation of first products; the dots denote all other terms with widening gap.
Their first product is
\begin{align}
    &(\bet_{E,2}[-1]\vac -\sum_{k\geq 5} \gam^{F,k-2}[0]\bet_{H,k}[-1] \vac +2 \sum_{k\geq 5}\gam^{H,k-2}[0]\bet_{E,k}[-1]\vac+\dots)_{(1)}\nn\\
    &\hspace{1cm}( -14\gam^{E,2}[-1]\vac +\sum_{k\geq 1} \gam^{E,k+2}[0]\bet_{H,k}[-1] \vac -2 \sum_{k\geq 1}\gam^{H,k+2}[0]\bet_{F,k}[-1]\vac +\dots)\nn\\
     =&(\bet_{E,2}[1] -\sum_{k\geq 5} \gam^{F,k-2}[1]\bet_{H,k}[0]  +2 \sum_{k\geq 5}\gam^{H,k-2}[-1]\bet_{E,k}[0]+\dots)\nn\\
    &\hspace{1cm}( -14\gam^{E,2}[-1]\vac +\sum_{k\geq 1} \gam^{E,k+2}[0]\bet_{H,k}[-1] \vac -2 \sum_{k\geq 1}\gam^{H,k+2}[0]\bet_{F,k}[-1]\vac +\dots)\nn\\
  =& (-14 z^2  - 4 \sum_{k\geq 5} z^{2k-2} )\vac +\dots = (-14 z^2  - 4 \frac{z^{8}}{1-z^2})\vac+\dots .
\end{align}
In the last line, the dots would correspond to terms which are not of the form $\CC(z)\vac$. However, the theorem above ensures that they give trivial contribution.
Applying the regularisation procedure to this expression, we obtain
\begin{align}
  & -14 z^2  - 4 \frac{z^{8}}{1-z^2} \leadsto -14e^{2y}-4\frac{e^{8y}}{1-e^{2y}} \nn\\
  &\qquad \leadsto -\frac{1}{2y}\Big(-14\frac{1+2y+\dots}{1+y+\dots} + 14\frac{1+4y+\dots}{1+y+\dots} -4\frac{1+8y+\dots}{1+y+\dots}\Big)\nn\\
  &\qquad \qquad = \frac{14}{2y}(1+2y+\dots)(1-y+\dots)-\frac{14}{2y}(1+4y+\dots)(1-y+\dots)\nn\\
  &\qquad \qquad \qquad -\frac{4}{2y}(1+8y+\dots)(1-y+\dots) \leadsto 0
\end{align}

\textbf{Remark.} It might be tempting to think that one could use this procedure to systematically regularise the products of the vertex Lie algebra $\hMn{1}$ as follows
\begin{equation}
  \label{eq:newprod}
  _{[i]}:=\reg \circ _{(i)} :\hMn1\times\hMn1 \to \hM[\leq 1],
\end{equation}
and conclude that \cref{prop1,prop2} define a homomorphism of vertex Lie algebras.
However, the space with these new products has \emph{not} the structure of a vertex Lie algebra, since Borcherds' identities are in general not satisfied.

\section{Proof of Theorem \ref{prop2}}\label{appB}
The proof of the theorem makes use of the doubling procedure introduced in \cite{analogFFR}.
For the sake of completeness, we will first recall the main ideas of that construction, which will be used below.

\subsection{The doubling trick}
  \label{doub}
  Recall from \cref{VAsplit} that it is not possible to lift the Lie algebra homomorphism at the vertex algebra level because first products are in general not well-defined.

  The so called \emph{doubling trick} was introduced in order to make sense of such products and construct a genuine homomorphism of vertex algebra.

  The idea is that the problem can be solved by suitably ``glueing together'' the algebra $\wt{\Der}\Oa(\n_+)$ with a ``negative copy'' of itself, $\wt{\Der}\Oa(\n_-)$ acting on the polynomial algebra $\Oa(\n_-)=\CC[X^{a,n}]_{(a,n)\in\sfA_-}$, with $\sfA_-:={(\alpha,0)}_{\alpha\in\mathring{\Delta}_-}\cup \mathcal I \times \ZZ_{\leq -1}$. With a construction analogous to the one outlined in the previous subsection, one can define $\Der \Oa(\n_-)$, its completion $\wt{\Der}\Oa(\n_-)$ and the subalgebra of elements with widening gap $\bDer\Oa(\n_-)$. Also, one defines the space $\Oa:=\CC[X^{a,n}]_{(a,n)\in\Ia\times\ZZ}$, and accordingly the completion $\wt{\Der}\Oa$ and the subalgebra of terms with widening gap $\bDer\Oa$.

  By using the involution map $\tau:\wt{\Der}\Oa\to\wt{\Der}\Oa$, with the property of exchanging $\wt{\Der}\Oa(\n_+)$ with $\wt{\Der}\Oa(\n_-)$ and vice-versa, one defines
  \begin{equation}
    \label{eq:doublemap}
    \ro:=\varrho + \tau\circ\varrho\circ\sigma:\g\to \wt{\Der}\Oa(\n_+)\oplus\wt{\Der}\Oa(\n_-)\hookrightarrow\wt{\Der}\Oa
  \end{equation}
  where $\sigma:\g\to\g$ is the Cartan involution, with the property of exchanging $\n_+$ with $\n_-$, namely
  \begin{equation}
    \label{eq:sigmadef}
    \sigma(J_{\alpha,n}) = J_{-\alpha,-n}, \qquad \sigma(J_i,n) = - J_{i,-n}.
  \end{equation}

  One can prove a similar result to \cref{th1}, now adapted to the doubled case:
  \begin{lemma}
    \label{lem2}
    For all $(a,n)\in\Ia\times\ZZ$
    \begin{equation}
      \label{eq:thWGdouble}
      \ro(J_{a,n})-\ff{a}{b}{c}\sum_{k\in\ZZ}X^{b,k-n}D_{c,k}:=\sum_{(b,m)\in\Ia\times\ZZ}R^{b,m}_{J_{a,n}}(X)D_{b,m} \in \bDer\Oa
    \end{equation}
    where $R\in\Oa$.
  \end{lemma}
  The most important feature is that the infinite quadratic sum now runs over all $k\in\ZZ$.
  Let us make an explicit example, in order to understand what this realisation looks like. Consider the element $J_{\alpha,1}\in\hsl_2$. It is realised as
  \begin{align}
    \label{eq:exrho}
    \ro(J_{\alpha,1}) = & \underbrace{ -\sum_{k\geq 3}X^{-\alpha,k-1}D_{1,k} +2\sum_{k\geq 3}X^{1,k-1}D_{\alpha,k} + \sum_{(b,m)\in\sfA}{}^{+}R^{b,m}_{\alpha,1}(X)D_{b,m}}_{\varrho(J_{\alpha,1})}\nn\\
                   &\underbrace{ -\sum_{k\leq -1}X^{-\alpha,k-1}D_{1,k} +2\sum_{k\leq -1}X^{1,k-1}D_{\alpha,k} +\sum_{(b,m)\in\sfA_-}{}^{-}R^{b,m}_{-\alpha,-1}(X)D_{b,m}}_{\tau\circ\varrho\circ\sigma(J_{\alpha,1})}\nn\\
    =& -\sum_{k\in\ZZ}X^{-\alpha,k-1}D_{1,k} +2\sum_{k\in\ZZ}X^{1,k-1}D_{\alpha,k} + \sum_{(b,m)\in\sfA}{}^{+}R^{b,m}_{\alpha,1}(X)D_{b,m}  \nn\\
                   &\qquad +\sum_{(b,m)\in\sfA_-}{}^{-}R^{b,m}_{-\alpha,-1}(X)D_{b,m} +\sum_{k=0}^{2}X^{-\alpha,k-1}D_{1,k} -2\sum_{k=0}^{2}X^{1,k-1}D_{\alpha,k}\nn\\
    =& -\sum_{k\in\ZZ}X^{-\alpha,k-1}D_{1,k} +2\sum_{k\in\ZZ}X^{1,k-1}D_{\alpha,k} + \sum_{(b,m)\in\Ia\times\ZZ}R^{b,m}_{\alpha,1}(X)D_{b,m}
  \end{align}
  where $\sum_{(b,m)\in\sfA_-}{}^{-}R^{b,m}_{-\alpha,-1}(X)D_{b,m} = \tau\circ\varrho\circ\sigma(\sum_{(b,m)\in\sfA}{}^{+}R^{b,m}_{\alpha,1}(X)D_{b,m})$. In the second-to-last step one ``fills the gap'' between the semi-infinite sums in the positive and negative directions, which is the reason for the appearance of a finite number of \emph{quadratic compensating terms}. The sum $\sum_{(b,m)\in\Ia\times\ZZ}R^{b,m}_{\alpha,1}(X)D_{b,m} = \sum_{(b,m)\in\sfA}{}^{+}R^{b,m}_{\alpha,1}(X)D_{b,m}  +\sum_{(b,m)\in\sfA_-}{}^{-}R^{b,m}_{-\alpha,-1}(X)D_{b,m}+$\textit{ finite quadratic compensating terms} in the last line is precisely the r.h.s. of \cref{eq:thWGdouble}.

  As before, $\wt{\Der}\Oa$ can be naturally embedded into the ``doubled'' Fock space of the $\bet\gam$-system $\tM_{\text{d}}$ and $\bDer\Oa$ into $\bM_{\text{d}}\subset \tM_{\text{d}}$, by the identification $X^{a,n} \mapsto \gam^{a,n}[0]$ and $D_{a,n}\mapsto \bet_{a,n}[-1]$.

  The main advantage of the glueing procedure is that one can now regard the infinite sums of quadratic terms $\sum_{k\in\ZZ}X^{a,k-n}D_{a,k}$ as new abstract generator  $\sfS^a_{b,n}$ of the loop algebra $\gl(\og)[t,t\inv]$, with commutation relations
  \begin{equation}
    \label{eq:Scommrel}
    [\sfS^a_{b,m},\sfS^c_{d,n}]=\delta^c_b\sfS^a_{d,n+m}-\delta^a_d\sfS^c_{b,n+m}.
  \end{equation}
  Let $\sfD$ be the derivation element for the homogeneous gradation of this algebra, obeying $[\sfD,\sfS^a_{b,n}]=n\sfS^a_{b,n}$. One has the homomorphism $\og\lp\to\gl(\og)\lp$, given by
  \begin{equation}
    \label{eq:Lad}
    J_{a,n} \mapsto \ff{a}{b}{c}\sfS^b_{c,n}, \qquad n\in\ZZ,
  \end{equation}
  where we are using the index summation convention for the Lie algebra indices.

  This can be extended to the whole affine algebra by declaring
  \begin{equation}
    \label{eq:Lad2}
    \sfk\mapsto 0, \qquad \sfd\mapsto \sfD.
  \end{equation}

  One can then introduce the loop algebra $\Ds=L(\gl(\og)[t,t\inv]\rtimes\CC\sfD)$, with generators $\sfS^a_{b,n}[N]$ and $\sfD[N]$, where $a,b\in\Ia$ and $n,N\in\ZZ$ with the following commutation relations
  \begin{align}
    \label{eq:Scomrel}
    \begin{gathered}
      [\sfS^a_{b,n}[N],\sfS^c_{d,m}[M]] = \delta^c_b\sfS^a_{d,n+m}[N+M]-\delta^a_d\sfS^c_{b,n+m}[N+M]\\
      [\sfD[N],\sfS^a_{b,n}[M]]=n\sfS^a_{b,n}[N+M]
    \end{gathered}
  \end{align}

  Introducing a vacuum vector $\vac$, one defines the vacuum Verma module $\VV_0^{\gl(\og)[t,t\inv]\rtimes\CC\sfD,0}$ over this loop algebra at level zero. The tensor product of $\Ds\ltimes\sfH$ modules,
  \begin{equation}
    \label{eq:defM}
    \overline{\mathsf{\bsf{M}}}:=\overline{\sfM}_\text{d}\otimes\VV_0^{\gl(\og)[t,t\inv]\rtimes\CC\sfD,0},
  \end{equation}
has the structure of a vertex algebra.

  As in the finite-type case, lifting the homomorphism of Lie algebras $\uprho$ to one of vertex algebras from $\VV_0^{\g,k}$ to $\bMM$, does not preserve the non-negative products and therefore does not define a homomorphism of vertex algebras. However, one has the doubled analogue of the map \cref{eq:phimap2}
  \begin{equation}
    \label{eq:phi}
    \upphi:\g\longrightarrow\Omega_{\Oa},
  \end{equation}
  where $\Omega_{\Oa}$ is the space of one forms $dX^{a,n}$, such that the map $\uprho + \upphi$ can be lifted to a homomorphism of vertex algebras
  \begin{equation}
    \label{eq:homVA}
    \tet:\VV_0^{\g,0}\longrightarrow\bMM,
  \end{equation}
  where the level of the vacuum Verma module has to be set to the very particular value $k=0$.

  \subsection{Undoubling}
  In order to prove our statement, we need to make contact between the double setting and the undoubled one.

There is the embedding map $p:\bMM\to\tM_{\text{d}}$, mapping the abstract generators to doubly infinite sums, namely
  \begin{align}
    \sfS\indices{^a_{b,n}}[-1]\vac &\xmapsto{\phantom{++}p\phantom{++}} \,\sum_{k\in\ZZ}\gam^{a,k-n}[0]\bet_{b,k}[-1]\vac,\\
    \sfD[-1]\vac \quad\, &\xmapsto{\phantom{++}p\phantom{++}} \,\sum_{k\in\ZZ} k \gam^{a,k}[0]\bet_{a,k}[-1]\vac,
  \end{align}
  and acting as the identity on the widening gap subspace $\bM_{\text{d}}\subset\bMM$. We can introduce the projectors onto the positive and negative subspaces
  \begin{align}
    \pi_+: \tM_{\text{d}} \to \tM(\n_+), \qquad \pi_-: \tM_{\text{d}} \to \tM(\n_-),
  \end{align}
  defined in the obvious way. However, recall that ``overlapping terms'' with both positive and negative loop modes, like $\gam^{a,1}[0]\bet_{b,-1}[-1]\vac$, are also well defined states in $\tM_{\text{d}}$. Hence, we also define $\pi_0:\tM_{\text{d}}\to\tM_{\text{d}}$ as follows
  \begin{align}\label{eq:pisum}
    \pi_0 := \id_{\tM_{\text{d}}} - \pi_+ - \pi_-.
  \end{align}

  We define the compositions $p_+:=\pi_+\circ p:\bMM\to\tM(\n_+)$,  $p_-:=\pi_-\circ p:\bMM\to\tM(\n_-)$ and  $p_0:=\pi_0\circ p:\bMM\to\tM_{\text{d}}$, and therefore we have
  \begin{align}\label{eq:deco}
    p = p_+ + p_- +p_0.
  \end{align}

In particular, by the definition of $\bMM$, we have $p_+(\bMM)\subseteq\hM_z(\n_+)$, because the elements in $\VV_0^{\gl(\og)[t,t\inv]\rtimes\CC\sfD,0}$ are mapped to semi-infinite sums that can always be regarded as elements of $\sfQ_z$ in \cref{eq:Qdef} setting $\alpha=0$, while elements in $\bM_{\text{d}}$ are mapped to $\bM_z(\n_+)$.
Recall from \cref{sec:Malgb} that this space is a vertex Lie algebras over $\CC(z)$ when regarding infinite sums as the expansion for small $z$ of some rational functions. Similarly, one can repeat similar arguments for the image $p_-(\bMM)\subseteq\hM_z(\n_-)$, finding that it is a vertex Lie algebra over $\CC(z\inv)$ when regarding infinite sums as the \emph{large}-$z$ expansions of rational functions. This allows us to employ the same regularisation procedure as described in \cref{sec:zeta} to regularise the products on both spaces.

In particular, we have the following result, which relates the values of regularised rational functions in $z$ and $z\inv$,
\begin{lemma}\label{inv}
  The regularisation map \textup{\texttt{reg}} is invariant under the inversion map $\varpi: z\mapsto z\inv$,  \ie $\textup{\texttt{reg}}[f(z)] = \textup{\texttt{reg}}[\varpi(f(z))]$, for any $f\in\CC(z)$.
\end{lemma}
\begin{proof}
We denote by $\{z_1,\dots,z_{n}\}\subset \CC $ the set of poles of $f\in\CC(z)$ and by $\{k_1,\dots,k_n\}$ their multiplicities. By partial fraction decomposition, we can write
  \begin{align}\label{PFD}
    f(z) = \sum_{i= 1}^n \frac{f_i(z)}{(z-z_i)^{k_i}}+f_0(z)
  \end{align}
  where $f_i\in\CC[z]$, $i=0,\dots,n$. Recall that by definition, the regularisation procedure is essentially the evaluation $z\to 1$, whenever this does not produce ill-defined quantities. Hence, if $z_i\neq 1$ one can explicitly evaluate the limit $z\to 1$. In this case, since $z=1$ is a fixed point for the map $\varpi$, the result will not change under inversion.

  Consider now the case when $z_i=1$ is one of the poles. We have
  \begin{align}
    \frac{1}{(z-1)^{k_i}},
  \end{align}
  where, without loss of generality, we set $f_i(z)=1$. Recall the expansion
  \begin{align}
      \frac{1}{1-e^y}=-\sum_{k\geq 0}\frac{B_k}{k!}y^{k-1},
  \end{align}
  where $B_k$ are the $k$-th Bernoulli numbers \cite{lepowski99}. By performing the steps \ref{reg1} - \ref{reg3} on this expression, we obtain
  \begin{gather}\label{eq:plus}
    \frac{1}{(z-1)^{k_i}} \leadsto \frac{(-1)^{k_i}}{(1-e^y)^{k_i}} = \left[\sum_{j\geq 0}\frac{B_j}{j!}y^{j-1}\right]^{k_i}:=y^{-k_i}\left[\sum_{j\geq 0}c_j y^{j}\right]^{k_i}.
  \end{gather}
 The constant term can be obtained by extracting the coefficient of the $k_i$-th power in $y$ of the series in brackets, which will have the following form
  \begin{align}\label{eq:CT}
\sum_{\{j_1,\dots,j_{k_i}\}}^{} c_{j_1} \cdots c_{j_{k_i}} y^{j_1+\dots+j_{k_i}},
  \end{align}
  where the sum is over all the tuples of reals $\{j_1,\dots,j_{k_i}\}$, with the constraint $\sum_{i=1}^{k_i}j_{i}=k_i$. For example, if $k_i=1$, the only tuple one can choose is $\{j_1\}=\{1\}$ and the constant term is just $c_1$. If $k_i=2$ the tuples are $\{j_1,j_2\} = \{ 1,1 \},\{0,2\},\{2,0\}$; plugging the values in \cref{eq:CT} we obtain $c_1^2+2c_0c_2$ as constant term.

Repeating analogous steps for the same expression where we first send $z\mapsto z\inv$, we get
  \begin{gather}\label{eq:minus}
    \frac{1}{(z\inv-1)^{k_i}} \leadsto \frac{(-1)^{k_i}}{(1-e^{-y})^{k_i}} = \left[\sum_{j\geq 0}\frac{(-1)^{j-1}B_j}{j!}y^{j-1}\right]^{k_i}:=y^{-k_i}\left[\sum_{j\geq 0}(-1)^{j-1}c_j y^{j}\right]^{k_i}.
  \end{gather}
and using the same notations as above the constant term will have the form
\begin{align}
  \sum_{\{j_1,\dots,j_{k_i}\}} (-1)^{j_1+\dots+j_{k_i}-k_i}c_{j_1} \cdots c_{j_{k_i}} y^{j_1+\dots+j_{k_i}}.
\end{align}
Since we have the constraint $\sum_{i=1}^{k_i}j_i=k_i$, the sign disappear. Therefore the constant terms from \cref{eq:plus} and \cref{eq:minus} agree. This concludes the proof.
\end{proof}

  We can now compute first product in the doubled setting, without abstract generators, using the regulation procedure, as defined in the following lemma
\begin{lemma}\label{lempi}
    For any $x,y\in\bMM[\leq 1]$, we have
    \begin{align}\label{eq:pisplit}
      p(x_{(1)}y) :=  \textup{\texttt{reg}}&[p_+(x)_{(1)}p_+(y)] + \textup{\texttt{reg}}[p_-(x)_{(1)}p_-(y)] + \lim_{z\to 1}[p_0(x)_{(1)} p_0(y)]
    \end{align}
where $\textup{\texttt{reg}}$ is the regularisation procedure introduced in \cref{reg}.
\end{lemma}

\begin{proof}
  We need to check all different possible combinations of products, namely when both $x$ and $y$ are in $\VV_0^{\gl(\og)[t,t\inv]\rtimes\CC\sfD,0}$, when they are both in $\bM_{\text{d}}[\leq 1]$ and the mixed case.

  Consider $\bM_{\text{d}}[\leq1]_{(1)}\bM_{\text{d}}[\leq1]$. In this case, the embedding map $p$ acts on this space as the identity, by definition. The vertex algebra products are all well-defined, since $\bM_{\text{d}}$ is a vertex subalgebra over $\CC$, hence the regularisation procedure is simply the evaluation at $z\to 1$. Therefore, the right hand side of \cref{eq:pisplit} is just the decomposition of such products relatively to the maps \eqref{eq:pisum}.

  For the mixed case, note that $\bM_{\text{d}}[\leq1 ]$ is a vertex algebra ideal in $\bMM[\leq1]$, and therefore $\VV_0^{\gl(\og)[t,t\inv]\rtimes\CC\sfD,0}{}_{(1)}\bM_{\text{d}}[\leq1]\subset\bM_{\text{d}}[\leq 1]$. Moreover, we have
  \begin{align}\label{eq:SMb}
    p(\sfS^a_{b,n}&[-1]\vac)_{(1)}p(\sum_{(c,m)\in\Ia\times\ZZ}P^{c,m}(\gam)\bet_{c,m}[-1]\vac)\nn\\
                    &=\sum_{k\in\ZZ}\gam^{a,k-n}[0]\bet_{b,k}[-1]\vac_{(1)}\sum_{(c,m)\in\Ia\times\ZZ}P^{c,m}(\gam)\bet_{c,m}[-1]\vac\nn\\
                    &=\sum_{k\in\ZZ}\wick{\c1\gam^{a,k-n}[1]\c2\bet_{b,k}[0]\sum_{(c,m)\in\Ia\times\ZZ}\c2P^{c,m}(\gam)\c1\bet_{c,m}[-1]\vac}\\
                    &=- \sum_{k\in\ZZ}\sum_{(c,m)\in\Ia\times\ZZ}z^{2k-n}\frac{\partial P^{c,m}(\gam)}{\partial\gam^{b,k}[0]}\delta^a_c\delta_{k-n,m}\vac=-\sum_{(c,m)\in\Ia\times\ZZ}\delta^a_cz^{2m+n}\frac{\partial P^{c,m}(\gam)}{\partial\gam^{b,n+m}[0]}\vac,\nn
  \end{align}

  which is precisely the same result one would obtain without embedding the two terms, considering the additional regulated commutation relations
  \begin{align}
    [S^a_{b,n}[N],\gam^{c,m}[M]] = z^{m} \delta^c_b\gam^{a,m-n}[N+M],\nn\\
    [S^a_{b,n}[N],\bet_{c,m}[M]] = -z^{m} \delta^a_c\bet_{b,m+n}[N+M].
  \end{align}
  Since the final expression in \cref{eq:SMb} only involves a finite number of non-zero terms, the regularisation procedure is just the evaluation $z\to 1$. As before, the right hand side of \cref{eq:pisplit} just follows from the decomposition \cref{eq:deco}.

  The only non-trivial check is when considering two elements in  $\VV_0^{\gl(\og)[t,t\inv]\rtimes\CC\sfD,0}$.
  In this case the left hand side of \cref{eq:pisplit} is identically zero by definition.
  First, we decompose each infinite sum, obtaining
  \begin{align}
    p(\sfS^a_{b,n}[-1]\vac) = &\sum_{k\in\ZZ}\gam^{a,k-n}[0]\bet_{b,k}[-1]\vac =\sum_{k\geq\max(0,n)}\gam^{a,k-n}[0]\bet_{b,k}[-1]\vac \nn\\
                                &\quad + \sum_{k\leq -\max(0,n)}\gam^{a,k-n}[0]\bet_{b,k}[-1]\vac +\sum_{k=-\max(0,n)+1}^{\max(0,n)+1}\gam^{a,k-n}[0]\bet_{b,k}[-1]\vac\nn\\
    =& p_+(\sfS^a_{b,n}[-1]\vac) + p_-(\sfS^a_{b,n}[-1]\vac)+p_0(\sfS^a_{b,n}[-1]\vac).\label{eq:Sdec}
  \end{align}
  where crucially the last sum is always finite. In the case $n=0$, one should instead consider $(\sum_{k > 0} + \sum_{k<0}+\delta_{k,0})\gam^{a,k}[0]\bet_{b,k}[-1]\vac$. This however does not alter the proof.
  Using \cref{eq:Sdec}, we can compute the regularised first products of two such states
  \begin{align}
   \reg[p_+(\sfS^a_{b,n}[-1]\vac)_{(1)}p_+(&\sfS^c_{d,m}[-1]\vac)] + \reg[p_-(\sfS^a_{b,n}[-1]\vac)_{(1)}p_-(\sfS^c_{d,m}[-1]\vac)]\nn\\&\qquad+ \reg[p_0(\sfS^a_{b,n}[-1]\vac)_{(1)}p_0(\sfS^c_{d,m}[-1]\vac)].
  \end{align}
  where no additional cross-terms appear, since they would only give trivial contractions. Since the image $p_0(\sfS^a_{b,n}[-1]\vac)$ is always a finite sum, the last term in the previous expression has only a finite number of non-zero contractions. For this reason, strictly speaking, the full regularisation procedure is not needed, as it just corresponds to the limit $z\to1$.

Writing out explicitly the images of the various projectors and computing the products, we obtain
  \begin{align}
 -\delta_{m+n,0}\delta^a_d\delta^c_b\left(\reg\Big[ \sum_{k\geq \max(0,n)}z^{2k-n}\Big]  +\reg\Big[\sum_{k\leq - \max(0,n)}z^{2k-n}\Big] +\lim_{z\to 1}\sum_{k= -\max(0,n)+1 }^{\max(0,n)-1}z^{2k-n} \right)
  \end{align}

  Regarding the infinite sums as expansions of rational functions, we obtain
  \begin{align}\label{eq:10}
    -\delta_{m+n,0}\delta^a_d\delta^c_b\left(\reg\Big[ \iota_{z=0} \frac{z^{2\max(0,n)-n}}{1-z^2}\Big] + \reg\Big[ \iota_{z=\infty} \frac{z^{-2\max(0,n)-n}}{1-z^{-2}}\Big]+\lim_{z\to1}\sum_{k= -\max(0,n)+1 }^{\max(0,n)-1}z^{2k-n}\right) .
  \end{align}
Finally, by explicitly performing the regularisation procedure, one finds
  \begin{align}
    \reg[p(\sfS^a_{b,n}&[-1]\vac)_{(1)}p(\sfS^c_{d,n}[-1]\vac)] \nn\\&= -\delta_{m+n,0}\delta^a_d\delta^c_b\Big( 1-2\max(0,n)+2\max(0,n)-1\Big)=0.
  \end{align}
\end{proof}

As a remark, note that the expression in \cref{eq:10} is zero on the nose, even \emph{without} employing the full regularisation procedure outlined in \cref{sec:zeta}. However, as we will see in the proof of the main theorem below, the remarkable cancellations only happen when we perform the other steps of the procedure.

\subsection{Proof of the main theorem}

\sloppy We now have all the necessary tools to prove the main theorem.
Consider the states  $J_{a,n}[-1]\vac, J_{b,m}[-1]\vac\in\VV_0^{\g,0}$. They can be mapped into $\bMM$, using the vertex algebra map $\tet$ in \cref{eq:homVA}. Since it is a homomorphism of vertex algebras, we have
    \begin{equation}
    \tet(J_{a,n}[-1]\vac)_{(1)}\tet(J_{b,m}[-1]\vac) = \tet(J_{a,n}[-1]\vac_{(1)}J_{b,m}[-1]\vac) = 0.
\end{equation}
Here the right-hand side is zero because it is defined from the vacuum Verma module at zero level, cfr. \cref{eq:homVA}. Acting on both sides with the map $p$ from \cref{eq:deco} we get
\begin{equation}\label{eq:thpr1}
    p\left(\tet(J_{a,n}[-1]\vac)_{(1)}\tet(J_{b,m}[-1]\vac)\right) = 0.
  \end{equation}

  Recall from the doubling construction summarised in \cref{doub} that the image of $\tet$ is obtained by glueing together a positive and a negative copy of the Lie algebra homomorphism \eqref{eq:liehom1}. For this reason, there are no overlapping terms, \ie
  \begin{align}
    p_0\circ \tet(x) = 0 \qquad \text{for all } x\in\VV_0^{\g,0}.
  \end{align}
  Moreover, for all $x\in\VV_0^{\g,0}$ we can identify
  \begin{align}
    p_+\circ \tet (x) = \vartheta(x), \qquad p_-\circ\tet (x) = \tau\circ\vartheta\circ\sigma (x)
  \end{align}
These facts can be understood looking at the first two lines of the example in \eqref{eq:exrho}.
Keeping this in mind, using the result of \cref{lempi} on \cref{eq:thpr1}, we find
\begin{equation}
  \textup{\texttt{reg}}[\vartheta(x)_{(1)}\vartheta(y)] + \textup{\texttt{reg}}[\tau\circ\vartheta\circ\sigma(x)_{(1)}\tau\circ\vartheta\circ\sigma(y)] = 0.
  \end{equation}
We will now proceed to show that these terms are in fact equal and therefore independently zero.

The first term will produce either a term proportional to the vacuum, when two quadratic semi-infinite sums are contracted together or when $(\bet_{a,n}[-1]\vac)_{(1)}=\bet_{a,n}[1]$ is contracted with a single $\gam^{a,n}[-1]\vac$, or terms of the form $R(\gam)\vac$, with $R\in\CC[\gam^{a,n}]_{(a,n)\in\sfA_+}$ of degree $>0$. Similarly, the second term, being the ``negative copy" coming from the glueing procedure, will produce either terms proportional to the vacuum or terms of the form $Q(\gam)\vac$, with $Q\in\CC[\gam^{a,n}]_{(a,n)\in\sfA_-}$ of degree $>0$.
It follows that the two contributions have to be independently zero, exception made for terms of the form $\CC(z)\vac$.

By direct calculation one finds that this product is proportional to the Killing form, and therefore it is non-trivial only in two cases: $(a,n)=(E_{\alpha},n)$, $(b,m)=(E_{-\alpha},-n)$ for some $\alpha_+\in\upcirc{\Delta}$, $n\in\ZZ$, and $(a,n)=(H_{i},n)$, $(b,m)=(H_{i},-n)$, where $H_i \in\upcirc{\h}$  $i=1,\dots,\rank\og$, is an orthogonal basis for the Cartan subalgebra.

In the first case, we can act with the involution $\sigma$ explicitly and use the symmetry property of the ${}_{(1)}$ product to get
\begin{align}
    \reg[\vartheta(J_{E_{\alpha},n}[-1]\vac)_{(1)}&\vartheta(J_{E_{-\alpha},-n}[-1]\vac)] \nn\\
    &+ \reg[\tau\circ\vartheta\circ\sigma(J_{E_{\alpha},n}[-1]\vac)_{(1)}\tau\circ\vartheta\circ\sigma(J_{E_{-\alpha},-n}[-1]\vac)] \nn\\
    =\reg[\vartheta(J_{E_{\alpha},n}[-1]\vac)_{(1)}&\vartheta(J_{E_{-\alpha},-n}[-1]\vac)]\nn \\
    & +\reg[\tau\circ\vartheta(J_{E_{-\alpha},-n}[-1]\vac)_{(1)}\tau\circ\vartheta(J_{E_{\alpha},n}[-1]\vac)]\nn\\
    =\reg[\vartheta(J_{E_{\alpha},n}[-1]\vac)_{(1)}&\vartheta(J_{E_{-\alpha},-n}[-1]\vac)]\nn \\
    & +\reg[\tau\circ\vartheta(J_{E_{\alpha},n}[-1]\vac)_{(1)}\tau\circ\vartheta(J_{E_{-\alpha},-n}[-1]\vac)],
\end{align}
where in the last step we have used the symmetry of the first product on the second term (cfr. \cref{eq:skew}).
In the second case, we similarly obtain
\begin{align}
    \reg[\vartheta(J_{H_i,n}[-1]\vac)_{(1)}&\vartheta(J_{H_i,-n}[-1]\vac) ]\nn\\
    &+ \reg[\tau\circ\vartheta\circ\sigma(J_{H_i,n}[-1]\vac)_{(1)}\tau\circ\vartheta\circ\sigma(J_{H_i,-n}[-1]\vac)] \nn\\
    =\reg[\vartheta(J_{H_i,n}[-1]\vac)_{(1)}&\vartheta(J_{H_i,-n}[-1]\vac)] \nn\\
    &+ \reg[\tau\circ\vartheta(J_{H_i,-n}[-1]\vac)_{(1)}\tau\circ\vartheta(J_{H_i,n}[-1]\vac)]\nn\\
    =\reg[\vartheta(J_{H_i,n}[-1]\vac)_{(1)}&\vartheta(J_{H_i,-n}[-1]\vac)] \nn\\
    &+\reg[\tau\circ\vartheta(J_{H_i,n}[-1]\vac)_{(1)}\tau\circ\vartheta(J_{H_i,-n}[-1]\vac)].
\end{align}
where in the last term we have used the symmetry of the ${}_{(1)}$ product.

By definition, the effect of the map $\tau$ is to ``mirror" the elements in $\hMn{1}$, to get a realisation of the algebra as derivations on $\Oa(\n_-)$. This means that whenever we have a non-trivial contraction between $a,b\in\hMn{1}$, we also have it between $\tau(a)$ and $\tau(b)$, with the only difference that $z$ is replaced with $ z\inv$, \ie
\begin{equation}
  \tau(a)_{(1)}\tau(b) = \varpi (a_{(1)}b)
\end{equation}
where $\varpi$ is the inversion map $z\mapsto z\inv$.

By \cref{inv}, we conclude that after regularisation, the two terms are equal and therefore they must be both independently zero.

\section{Conclusion}

In this paper we rigorously proved an observation made in \cite{analogFFR}.
To do that, we first defined the space $\hM_z$, which has the structure of a vertex Lie algebra, whose products depend on a regulation parameter $z$. It is known that the unregulated limit $z\to 1$ gives rise to ill-defined quantities. For this reason, we introduced a suitable regularisation procedure, thanks to which these products were ``renormalised''.
Using this, we showed that the regularised first products on the image of the map $\vartheta$ are not just well-defined, but remarkably they are zero on the nose.

\section*{Acknowledgements}
The author wants to thank Charles Young and D. Simon H. Jonsson for useful discussions.
The author gratefully acknowledges the financial support of the grant ERC-2022-CoG - FAIM 101088193.

\appendix
\numberwithin{equation}{section}

\section{Realisation of $\hsl_3$ as differential operators}
Explicit expression for the images of $J_{\alpha_1,n}$, with $n=-1,0,1 $ through the Lie algebra homomorphism \eqref{eq:liehom1}, truncated up to loop order 3,4 or 5 (depending on the length of the expressions).

\begin{align}
     \varrho&(J_{{\alpha_1},1}) = D_{{\alpha_1},1}-\sum_{k\geq 2} X^{{\alpha_2},k-1}D_{{\alpha_1+\alpha_2},k}-\sum_{k\geq 3} X^{{-\alpha_1},k-1}D_{1,k} +\sum_{k\geq 3} X^{{-\alpha_1-\alpha_2},k-1}D_{{-\alpha_2},k}\nn\\
       & +2\sum_{k\geq 3}X^{1,k-1}D_{{\alpha_1},k} -\sum_{k\geq 3}X^{2,k-1}D_{{\alpha_1},k}\nn\\
       &+(- X^{{-\alpha_2},2}X^{{\alpha_2},1}+\dots) D_{{\alpha_1},4} +(- X^{{-\alpha_1-\alpha_2},2}X^{{\alpha_2},1}+\dots) D_{1,4}\nn\\
      & + (- X^{{-\alpha_1-\alpha_2},2}X^{{\alpha_2},1}+\dots)D_{2,4} +( X^{1,2}X^{{\alpha_2},1}+ X^{2,2}X^{{\alpha_2},1} + \dots) D_{{\alpha_1+\alpha_2},4}\nn\\
      & +(- 2X^{{-\alpha_1},2}X^{{\alpha_1},2}- X^{{-\alpha_1-\alpha_2},2}X^{{\alpha_1+\alpha_2},2} \nn\\
      &\qquad \qquad - 2X^{1,2}X^{1,2} + 2X^{2,2}X^{1,2} - \frac{1}{2}X^{2,2}X^{2,2}+\dots)D_{{\alpha_1},5}\nn\\
     &  + (X^{{-\alpha_1},2}X^{{\alpha_2},2}+\dots)D_{{\alpha_2},5}\nn\\
            &  +( - X^{{-\alpha_1},2}X
              ^{{\alpha_1+\alpha_2},2}+ 2X^{1,2}X^{{\alpha_2},2} - X^{2,2}X^{{\alpha_2},2}+\dots)D_{{\alpha_1+\alpha_2},5}\nn\\
      & + (X^{{-\alpha_1},2}X^{{-\alpha_1},2}+\dots)D_{{-\alpha_1},5}\nn\\
      & +(- X^{{-\alpha_2},2}X^{{-\alpha_1},2}+- X^{1,2}X^{{-\alpha_1-\alpha_2},2}  + 2X^{2,2}X^{{-\alpha_1-\alpha_2},2} +\dots)D_{{-\alpha_2},5}\nn\\
       &- (X^{{-\alpha_1-\alpha_2},2}X^{{\alpha_2},2}+\dots)D_{2,5}+ (X^{{-\alpha_1-\alpha_2},2}X^{{-\alpha_1},2}+\dots)D_{{-\alpha_1-\alpha_2},5}+\dots
\end{align}

\begin{align}
  \varrho&(J_{{\alpha_1},0}) = D_{{\alpha_1},0}-\sum_{k\geq 0} X^{{\alpha_2},k}D_{{\alpha_1+\alpha_2},k}-\sum_{k\geq 1} X^{{-\alpha_1},k}D_{1,k} +\sum_{k\geq 1} X^{{-\alpha_1-\alpha_2},k}D_{{-\alpha_2},k}\nn\\
         & +2\sum_{k\geq 1}X^{1,k}D_{{\alpha_1},k} -\sum_{k\geq 1}X^{2,k}D_{{\alpha_1},k}\nn\\
         &+(- 2 X^{1,{1}} X^{1,{1}}+ 2 X^{2,{1}} X^{1,{1}} - \frac12 X^{2,{1}} X^{2,{1}} D_{\alpha_1,{2}}+\dots)D_{\alpha_1,{2}} +(X^{-\alpha_1,{1}} X^{-\alpha_1,{1}}+\dots) D_{-\alpha_1,{2}}\nn\\
         &+(- X^{-\alpha_2,{1}} X^{-\alpha_1,{1}}+\dots) D_{-\alpha_2,{2}} +( X^{-\alpha_1-\alpha_2,{1}} X^{-\alpha_1,{1}} +\dots)D_{-\alpha_1-\alpha_2,{2}}\nn\\
         & +(X^{-\alpha_2,{1}} X^{-\alpha_1,{1}} X^{\alpha_1+\alpha_2,{1}} + \frac43 X^{1,{1}} X^{1,{1}} X^{1,{1}}- 2 X^{2,{1}} X^{1,{1}} X^{1,{1}}\nn\\&\qquad \qquad + X^{2,{1}} X^{2,{1}} X^{1,{1}}- \frac16 X^{2,{1}} X^{2,{1}} X^{2,{1}} D_{\alpha_1,{3}}+\dots) D_{\alpha_1,{3}}\nn\\
         & +(X^{-\alpha_1,{1}} X^{-\alpha_1,{1}} X^{\alpha_1+\alpha_2,{1}} +\dots )D_{\alpha_2,{3}} \nn\\
         &+(- 2 X^{1,{1}} X^{1,{1}} X^{\alpha_2,{1}}+ 2 X^{2,{1}} X^{1,{1}} X^{\alpha_2,{1}}- \frac12 X^{2,{1}} X^{2,{1}} X^{\alpha_2,{1}}+\dots ) D_{\alpha_1+\alpha_2,{3}}\nn\\
         &+( - X^{-\alpha_1-\alpha_2,{1}} X^{-\alpha_1,{1}} X^{\alpha_2,{1}} + 2 X^{1,{1}} X^{-\alpha_1,{1}} X^{-\alpha_1,{1}} - X^{2,{1}} X^{-\alpha_1,{1}} X^{-\alpha_1,{1}}+\dots)D_{-\alpha_1,{3}}\nn\\
         & +(X^{-\alpha_1-\alpha_2,{1}} X^{-\alpha_1,{1}} X^{\alpha_1,{1}} + X^{1,{1}} X^{-\alpha_2,{1}} X^{-\alpha_1,{1}} - 2 X^{2,{1}} X^{-\alpha_2,{1}} X^{-\alpha_1,{1}}+\dots) D_{-\alpha_2,{3}}\nn\\
         &+(- X^{-\alpha_2,{1}} X^{-\alpha_1,{1}} X^{-\alpha_1,{1}} + X^{1,{1}} X^{-\alpha_1-\alpha_2,{1}} X^{-\alpha_1,{1}} + X^{2,{1}} X^{-\alpha_1-\alpha_2,{1}} X^{-\alpha_1,{1}} +\dots) D_{-\alpha_1-\alpha_2,{3}}\nn\\
         & + (- X^{-\alpha_1,{1}} X^{-\alpha_1,{1}} X^{\alpha_1,{1}} - X^{-\alpha_1-\alpha_2,{1}} X^{-\alpha_1,{1}} X^{\alpha_1+\alpha_2,{1}}\dots )D_{1,3}\nn\\
         &+ (X^{-\alpha_2,{1}} X^{-\alpha_1,{1}} X^{\alpha_2,{1}}- X^{-\alpha_1-\alpha_2,{1}} X^{-\alpha_1,{1}} X^{\alpha_1+\alpha_2,{1}}+\dots) D_{2,{3}}+\dots
\end{align}

\begin{align}
  \varrho&(J_{{\alpha_1},-1}) = -\sum_{k\geq 0} X^{{\alpha_2},k+1}D_{{\alpha_1+\alpha_2},k}-\sum_{k\geq 1} X^{{-\alpha_1},k+1}D_{1,k} +\sum_{k\geq 1} X^{{-\alpha_1-\alpha_2},k+1}D_{{-\alpha_2},k}\nn\\
         & +2\sum_{k\geq 0}X^{1,k+1}D_{{\alpha_1},k} -\sum_{k\geq 0}X^{2,k+1}D_{{\alpha_1},k}\nn\\
         &  + ( 2X^{-\alpha_1,{1}} X^{\alpha_1,{0}} + X^{-\alpha_1-\alpha_2,{1}} X^{\alpha_2,{0}} X^{\alpha_1,{0}}+ X^{-\alpha_1-\alpha_2,{1}} X^{\alpha_1+\alpha_2,{0}} +\dots )D_{\alpha_1,{0}}\nn\\
         &+ (- X^{-\alpha_1,{1}} X^{\alpha_2,{0}}- X^{-\alpha_1-\alpha_2,{1}} X^{\alpha_2,{0}} X^{\alpha_2,{0}}  +\dots) D_{\alpha_2,{0}}\nn\\
         &+(X^{-\alpha_1,{1}} X^{\alpha_1+\alpha_2,{0}} - X^{-\alpha_1-\alpha_2,{1}} X^{\alpha_1+\alpha_2,{0}} X^{\alpha_2,{0}} D_{\alpha_1+\alpha_2,{0}} - 2 X^{1,{1}} X^{\alpha_2,{0}} + X^{2,{1}} X^{\alpha_2,{0}}+\dots )D_{\alpha_1+\alpha_2,{0}}\nn\\
         &+(X^{-\alpha_2,{1}} X^{\alpha_2,{1}} + 2 X^{1,{1}} X^{1,{1}}  - 2 X^{2,{1}} X^{1,{1}}+\frac12 X^{2,{1}} X^{2,{1}} D_{\alpha_1,{1}} +\dots)D_{\alpha_1,{1}}\nn\\
         &+(- X^{-\alpha_1,{1}} X^{\alpha_2,{1}} +\dots )D_{\alpha_2,{1}} + (- X^{1,{1}} X^{\alpha_2,{1}} X^{2,{1}} X^{\alpha_2,{1}}+\dots )D_{\alpha_1+\alpha_2,{1}}\nn\\
         &+(- X^{-\alpha_1,{1}} X^{-\alpha_1,{1}}+\dots)  D_{-\alpha_1,{1}} +(2 X^{1,{1}} X^{-\alpha_1-\alpha_2,{1}} - X^{2,{1}} X^{-\alpha_1-\alpha_2,{1}} +\dots )D_{-\alpha_2,{1}}\nn\\
         &+(- X^{-\alpha_1-\alpha_2,{1}} X^{-\alpha_1,{1}}+\dots) D_{-\alpha_1-\alpha_2,{1}} +(X^{-\alpha_1-\alpha_2,{1}} X^{\alpha_2,{1}}- 2 X^{1,{1}} X^{-\alpha_1,{1}} + X^{2,{1}} X^{-\alpha_1,{1}}+\dots) D_{1,{1}}\nn\\
         &+( X^{-\alpha_1-\alpha_2,{1}} X^{\alpha_2,{1}} +\dots)D_{2,{1}}\nn\\
         &+( 2 X^{-\alpha_1,{2}} X^{\alpha_1,{1}}+ X^{-\alpha_1-\alpha_2,{2}} X^{\alpha_1+\alpha_2,{1}} - 2 X^{1,{1}} X^{-\alpha_2,{1}} X^{\alpha_2,{1}} - \frac83 X^{1,{1}} X^{1,{1}} X^{1,{1}}\nn\\
         &\qquad + X^{2,{1}} X^{-\alpha_2,{1}} X^{\alpha_2,{1}}+ 4 X^{2,{1}} X^{1,{1}} X^{1,{1}} - 2 X^{2,{1}} X^{2,{1}} X^{1,{1}} + \frac13 X^{2,{1}} X^{2,{1}} X^{2,{1}} D_{\alpha_1,{2}} +\dots)D_{\alpha_1,{2}}\nn\\
         &+(- X^{-\alpha_1,{2}} X^{\alpha_2,{1}}- X^{1,{1}} X^{-\alpha_1,{1}} X^{\alpha_2,{1}}+ 2 X^{2,{1}} X^{-\alpha_1,{1}} X^{\alpha_2,{1}}+\dots)D_{\alpha_2,{2}}\nn\\
         &+( X^{-\alpha_1,{1}} X^{\alpha_2,{1}} X^{\alpha_1,{1}} + X^{-\alpha_1,{2}} X^{\alpha_1+\alpha_2,{1}}+ \frac12 X^{1,{1}} X^{1,{1}} X^{\alpha_2,{1}}\nn\\
         &\qquad- 2 X^{1,{2}} X^{\alpha_2,{1}} + X^{2,{1}} X^{1,{1}} X^{\alpha_2,{1}} + \frac12 X^{2,{1}} X^{2,{1}} X^{\alpha_2,{1}}+ X^{2,{2}} X^{\alpha_2,{1}} D_{\alpha_1+\alpha_2,{2}}+\dots) D_{\alpha_1+\alpha_2,{2}}\nn\\
         &+(- X^{-\alpha_1-\alpha_2,{1}} X^{-\alpha_1,{1}} X^{\alpha_2,{1}} + 2 X^{1,{1}} X^{-\alpha_1,{1}} X^{-\alpha_1,{1}}- X^{2,{1}} X^{-\alpha_1,{1}} X^{-\alpha_1,{1}} +\dots)D_{-\alpha_1,{2}}\nn\\
         &+(- X^{-\alpha_1-\alpha_2,{1}} X^{-\alpha_2,{1}} X^{\alpha_2,{1}}- 2 X^{1,{1}} X^{-\alpha_2,{1}} X^{-\alpha_1,{1}}+ X^{1,{1}} X^{-\alpha_1-\alpha_2,{2}}\nn\\&\qquad+ X^{2,{1}} X^{-\alpha_2,{1}} X^{-\alpha_1,{1}}- 2 X^{2,{1}} X^{-\alpha_1-\alpha_2,{2}} D_{-\alpha_2,{2}}+\dots) D_{-\alpha_2,{2}}\nn\\
         &+(- X^{-\alpha_1-\alpha_2,{1}} X^{-\alpha_1-\alpha_2,{1}} X^{\alpha_2,{1}} + 2 X^{1,{1}} X^{-\alpha_1-\alpha_2,{1}} X^{-\alpha_1,{1}} \nn\\&\qquad- X^{2,{1}} X^{-\alpha_1-\alpha_2,{1}} X^{-\alpha_1,{1}} D_{-\alpha_1-\alpha_2,{2}}+\dots)D_{-\alpha_1-\alpha_2,{2}}\nn\\
         &+(- X^{-\alpha_2,{1}} X^{-\alpha_1,{1}} X^{\alpha_2,{1}} + X^{-\alpha_1-\alpha_2,{2}} X^{\alpha_2,{1}} +\dots) D_{2,{2}}\nn\\
         &+(X^{-\alpha_2,{1}} X^{-\alpha_1,{1}} X^{\alpha_2,{1}} X^{\alpha_1,{1}}+ X^{-\alpha_2,{2}} X^{\alpha_2,{2}} + X^{-\alpha_1-\alpha_2,{1}} X^{-\alpha_2,{1}} X^{\alpha_1+\alpha_2,{1}} X^{\alpha_2,{1}}\nn\\
         &\qquad+ 2 X^{1,{1}} X^{-\alpha_2,{1}} X^{-\alpha_1,{1}} X^{\alpha_1+\alpha_2,{1}} - X^{1,{1}} X^{-\alpha_1-\alpha_2,{2}} X^{\alpha_1+\alpha_2,{1}}+ 2 X^{1,{1}} X^{1,{1}} X^{-\alpha_2,{1}} X^{\alpha_2,{1}}\nn\\
         &\qquad+ 2 X^{1,{1}} X^{1,{1}} X^{1,{1}} X^{1,{1}} + 2 X^{1,{2}} X^{1,{2}}- X^{2,{1}} X^{-\alpha_2,{1}} X^{-\alpha_1,{1}} X^{\alpha_1+\alpha_2,{1}}+ 2 X^{2,{1}} X^{-\alpha_1-\alpha_2,{2}} X^{\alpha_1+\alpha_2,{1}}\nn\\
         &\qquad- 2 X^{2,{1}} X^{1,{1}} X^{-\alpha_2,{1}} X^{\alpha_2,{1}} - 4 X^{2,{1}} X^{1,{1}} X^{1,{1}} X^{1,{1}}+ \frac12 X^{2,{1}} X^{2,{1}} X^{-\alpha_2,{1}} X^{\alpha_2,{1}}\nn\\
         &\qquad+ 3 X^{2,{1}} X^{2,{1}} X^{1,{1}} X^{1,{1}}- X^{2,{1}} X^{2,{1}} X^{2,{1}} X^{1,{1}}+ \frac18 X^{2,{1}} X^{2,{1}} X^{2,{1}} X^{2,{1}}\nn\\
         &\qquad- 2 X^{2,{2}} X^{1,{2}}+\frac12 X^{2,{2}} X^{2,{2}} +\dots) D_{\alpha_1,{3}}\nn\\
         &+(- X^{-\alpha_1,{2}} X^{\alpha_2,{2}}- 2 X^{-\alpha_2,{1}} X^{-\alpha_1,{1}} X^{\alpha_2,{1}} X^{\alpha_2,{1}}- X^{-\alpha_1-\alpha_2,{1}} X^{-\alpha_1,{1}} X^{\alpha_1+\alpha_2,{1}} X^{\alpha_2,{1}}\nn\\
         &\qquad+ X^{-\alpha_1-\alpha_2,{2}} X^{\alpha_2,{1}} X^{\alpha_2,{1}}+ 2 X^{1,{1}} X^{-\alpha_1,{1}} X^{-\alpha_1,{1}} X^{\alpha_1+\alpha_2,{1}}- \frac12 X^{1,{1}} X^{1,{1}} X^{-\alpha_1,{1}} X^{\alpha_2,{1}} \nn\\
         &\qquad- X^{2,{1}} X^{-\alpha_1,{1}} X^{-\alpha_1,{1}} X^{\alpha_1+\alpha_2,{1}}+ 2 X^{2,{1}} X^{1,{1}} X^{-\alpha_1,{1}} X^{\alpha_2,{1}}- 2 X^{2,{1}} X^{2,{1}} X^{-\alpha_1,{1}} X^{\alpha_2,{1}}+\dots) D_{\alpha_2,{3}}\nn\\
         &+(2 X^{-\alpha_1,{2}} X^{\alpha_2,{1}} X^{\alpha_1,{1}}- X^{-\alpha_2,{1}} X^{-\alpha_1,{1}} X^{\alpha_1+\alpha_2,{1}} X^{\alpha_2,{1}} + X^{-\alpha_1-\alpha_2,{2}} X^{\alpha_1+\alpha_2,{1}} X^{\alpha_2,{1}} \nn\\
         &\qquad+ X^{1,{1}} X^{-\alpha_1,{1}} X^{\alpha_2,{1}} X^{\alpha_1,{1}}- 2 X^{1,{1}} X^{-\alpha_2,{1}} X^{\alpha_2,{1}} X^{\alpha_2,{1}}- \frac{17}{6} X^{1,{1}} X^{1,{1}} X^{1,{1}} X^{\alpha_2,{1}} \nn\\
         &\qquad- X^{1,{2}} X^{\alpha_2,{2}}- 2 X^{2,{1}} X^{-\alpha_1,{1}} X^{\alpha_2,{1}} X^{\alpha_1,{1}}+ X^{2,{1}} X^{-\alpha_2,{1}} X^{\alpha_2,{1}} X^{\alpha_2,{1}}+ \frac72 X^{2,{1}} X^{1,{1}} X^{1,{1}} X^{\alpha_2,{1}}\nn\\
         &\qquad - \frac52 X^{2,{1}} X^{2,{1}} X^{1,{1}} X^{\alpha_2,{1}} + \frac16 X^{2,{1}} X^{2,{1}} X^{2,{1}} X^{\alpha_2,{1}}- X^{2,{2}} X^{\alpha_2,{2}} +\dots)D_{\alpha_1+\alpha_2,{3}}\nn\\
         &+(- X^{-\alpha_1,{2}} X^{-\alpha_1,{2}}+ X^{-\alpha_1-\alpha_2,{1}} X^{-\alpha_1-\alpha_2,{1}} X^{\alpha_2,{1}} X^{\alpha_2,{1}} - 4 X^{1,{1}} X^{-\alpha_1-\alpha_2,{1}} X^{-\alpha_1,{1}} X^{\alpha_2,{1}} \nn\\
         &\qquad+ 4 X^{1,{1}} X^{1,{1}} X^{-\alpha_1,{1}} X^{-\alpha_1,{1}}+ 2 X^{2,{1}} X^{-\alpha_1-\alpha_2,{1}} X^{-\alpha_1,{1}} X^{\alpha_2,{1}}- 4 X^{2,{1}} X^{1,{1}} X^{-\alpha_1,{1}} X^{-\alpha_1,{1}} \nn\\
         &\qquad+ X^{2,{1}} X^{2,{1}} X^{-\alpha_1,{1}} X^{-\alpha_1,{1}}+\dots) D_{-\alpha_1,{3}}\nn\\
         &+( X^{-\alpha_2,{1}} X^{-\alpha_2,{1}} X^{-\alpha_1,{1}} X^{\alpha_2,{1}} - X^{-\alpha_1-\alpha_2,{1}} X^{-\alpha_1-\alpha_2,{1}} X^{\alpha_2,{1}} X^{\alpha_1,{1}} \nn\\
         &\qquad+ 2 X^{1,{1}} X^{-\alpha_1-\alpha_2,{1}} X^{-\alpha_1,{1}} X^{\alpha_1,{1}}+ X^{1,{1}} X^{-\alpha_1-\alpha_2,{1}} X^{-\alpha_2,{1}} X^{\alpha_2,{1}}+ 2 X^{1,{1}} X^{1,{1}} X^{-\alpha_2,{1}} X^{-\alpha_1,{1}}\nn\\
         &\qquad- \frac12 X^{1,{1}} X^{1,{1}} X^{-\alpha_1-\alpha_2,{2}} + 2 X^{1,{2}} X^{-\alpha_1-\alpha_2,{2}}- X^{2,{1}} X^{-\alpha_1-\alpha_2,{1}} X^{-\alpha_1,{1}} X^{\alpha_1,{1}}\nn\\
         &\qquad- 2 X^{2,{1}} X^{-\alpha_1-\alpha_2,{1}} X^{-\alpha_2,{1}} X^{\alpha_2,{1}} - 5 X^{2,{1}} X^{1,{1}} X^{-\alpha_2,{1}} X^{-\alpha_1,{1}}+ 2 X^{2,{1}} X^{1,{1}} X^{-\alpha_1-\alpha_2,{2}} +\dots)D_{-\alpha_2,{3}}\nn\\
         &+(X^{-\alpha_1-\alpha_2,{1}} X^{-\alpha_2,{1}} X^{-\alpha_1,{1}} X^{\alpha_2,{1}}- X^{-\alpha_1-\alpha_2,{2}} X^{-\alpha_1,{2}}- 2 X^{1,{1}} X^{-\alpha_2,{1}} X^{-\alpha_1,{1}} X^{-\alpha_1,{1}}\nn\\
         &\qquad- X^{1,{1}} X^{-\alpha_1-\alpha_2,{1}} X^{-\alpha_1-\alpha_2,{1}} X^{\alpha_2,{1}}+ 2 X^{1,{1}} X^{1,{1}} X^{-\alpha_1-\alpha_2,{1}} X^{-\alpha_1,{1}}+ X^{2,{1}} X^{-\alpha_2,{1}} X^{-\alpha_1,{1}} X^{-\alpha_1,{1}}\nn\\
         &\qquad- X^{2,{1}} X^{-\alpha_1-\alpha_2,{1}} X^{-\alpha_1-\alpha_2,{1}} X^{\alpha_2,{1}}+ X^{2,{1}} X^{1,{1}} X^{-\alpha_1-\alpha_2,{1}} X^{-\alpha_1,{1}}\nn\\
         &\qquad- X^{2,{1}} X^{2,{1}} X^{-\alpha_1-\alpha_2,{1}} X^{-\alpha_1,{1}} +\dots)D_{-\alpha_1-\alpha_2,{3}}\nn\\
         &+(X^{-\alpha_1-\alpha_2,{1}} X^{-\alpha_1,{1}} X^{\alpha_2,{1}} X^{\alpha_1,{1}}+ X^{-\alpha_1-\alpha_2,{1}} X^{-\alpha_1-\alpha_2,{1}} X^{\alpha_1+\alpha_2,{1}} X^{\alpha_2,{1}}+ X^{-\alpha_1-\alpha_2,{2}} X^{\alpha_2,{2}}\nn\\
         &\qquad- 2 X^{1,{1}} X^{-\alpha_1,{1}} X^{-\alpha_1,{1}} X^{\alpha_1,{1}}- 2 X^{1,{1}} X^{-\alpha_1-\alpha_2,{1}} X^{-\alpha_1,{1}} X^{\alpha_1+\alpha_2,{1}} - 2 X^{1,{2}} X^{-\alpha_1,{2}}\nn\\
         &\qquad + X^{2,{1}} X^{-\alpha_1,{1}} X^{-\alpha_1,{1}} X^{\alpha_1,{1}}+ X^{2,{1}} X^{-\alpha_1-\alpha_2,{1}} X^{-\alpha_1,{1}} X^{\alpha_1+\alpha_2,{1}} + X^{2,{2}} X^{-\alpha_1,{2}} +\dots)D_{1,{3}} \nn\\
         &+( X^{-\alpha_1-\alpha_2,{1}} X^{-\alpha_2,{1}} X^{\alpha_2,{1}} X^{\alpha_2,{1}}+ X^{-\alpha_1-\alpha_2,{1}} X^{-\alpha_1-\alpha_2,{1}} X^{\alpha_1+\alpha_2,{1}} X^{\alpha_2,{1}}\nn\\
         &\qquad+ X^{-\alpha_1-\alpha_2,{2}} X^{\alpha_2,{2}} + 2 X^{1,{1}} X^{-\alpha_2,{1}} X^{-\alpha_1,{1}} X^{\alpha_2,{1}}- 2 X^{1,{1}} X^{-\alpha_1-\alpha_2,{1}} X^{-\alpha_1,{1}} X^{\alpha_1+\alpha_2,{1}}\nn\\
         &\qquad- X^{1,{1}} X^{-\alpha_1-\alpha_2,{2}} X^{\alpha_2,{1}}- X^{2,{1}} X^{-\alpha_2,{1}} X^{-\alpha_1,{1}} X^{\alpha_2,{1}}+ X^{2,{1}} X^{-\alpha_1-\alpha_2,{1}} X^{-\alpha_1,{1}} X^{\alpha_1+\alpha_2,{1}}\nn\\
         &\qquad+ 2 X^{2,{1}} X^{-\alpha_1-\alpha_2,{2}} X^{\alpha_2,{1}} + 2 X^{2,{1}} X^{2,{1}} X^{-\alpha_2,{1}} X^{-\alpha_1,{1}}\nn\\
         &\qquad- 2 X^{2,{1}} X^{2,{1}} X^{-\alpha_1-\alpha_2,{2}}- X^{2,{2}} X^{-\alpha_1-\alpha_2,{2}} +\dots) D_{2,3}+\dots
\end{align}

For the sake of completeness, we also report an example of the image of a generator from $\n_-$, truncated at order 4.
\begin{align}
\varrho&(J_{{-\alpha_1,1}})= D_{{-\alpha_1},1}  +\sum_{k\geq 2}X^{{\alpha_1},k-1}D_{{1},k} - \sum_{k\geq 3}X^{{\alpha_1+\alpha_2},k-1}D_{{\alpha_2},k}\nn\\
&+ \sum_{k\geq 2}X^{{-\alpha_2},k-1}D_{{-\alpha_1-\alpha_2},k} -2\sum_{k\geq 2}X^{{1},k-1}D_{{-\alpha_1},k} +\sum_{k\geq 2}X^{{2},k-1}D_{{-\alpha_1},k}\nn\\
& + ( X^{\alpha_1,1}X^{\alpha_1,1} +\dots ) D_{\alpha_1,3} +( - X^{\alpha_2,1}X^{\alpha_1,1} - 2X^{1,1}X^{\alpha_1+\alpha_2,1}
 + X^{2,1}X^{\alpha_1+\alpha_2,1} + \dots) D_{\alpha_2,3}\nn\\
 &+ (- X^{-\alpha_2,1}X^{\alpha_2,1}  - 2X^{1,1}X^{1,1} + 2X^{2,1}X^{1,1} - \frac12 X^{2,1}X^{2,1} +\dots )D_{-\alpha_1,3}\nn\\
 &+ ( X^{-\alpha_2,1}X^{\alpha_1,1} +\dots  )D_{-\alpha_2,3} +(X^{\alpha_1+\alpha_2,1}X^{\alpha_1,1} +\dots ) D_{\alpha_1+\alpha_2,3}\nn\\
 &+ (X^{1,1}X^{-\alpha_2,1} + X^{2,1}X^{-\alpha_2,1} +\dots)D_{-\alpha_1-\alpha_2,3}\nn\\
 &+(- X^{-\alpha_2,1}X^{\alpha_1+\alpha_2,1} + 2X^{1,1}X^{\alpha_1,1} - X^{2,1}X^{\alpha_1,1}+\dots )D_{1,3} + (- X^{-\alpha_2,1}X^{\alpha_1+\alpha_2,1} +\dots )D_{2,3}\nn\\
& +(- X^{-\alpha_2,1}X^{\alpha_1+\alpha_2,1}X^{\alpha_1,1} + 2X^{1,1}X^{\alpha_1,1}X^{\alpha_1,1}  - X^{2,1}X^{\alpha_1,1}X^{\alpha_1,1} +\dots)D_{\alpha_1,4}\nn\\
&+( - X^{-\alpha_2,1}X^{\alpha_1+\alpha_2,1}X^{\alpha_2,1} - 2X^{1,1}X^{\alpha_2,1}X^{\alpha_1,1} - 2X^{1,1}X^{1,1}X^{\alpha_1+\alpha_2,1} - X^{1,2}X^{\alpha_1+\alpha_2,1} \nn\\&\qquad + X^{2,1}X^{\alpha_2,1}X^{\alpha_1,1} + 2X^{2,1}X^{1,1}X^{\alpha_1+\alpha_2,1} - \frac12 X^{2,1}X^{2,1}X^{\alpha_1+\alpha_2,1} + 2X^{2,2}X^{\alpha_1+\alpha_2,1} +\dots) D_{\alpha_2,4}\nn\\
&+(- 2X^{-\alpha_1,2}X^{\alpha_1,1}  - X^{-\alpha_1-\alpha_2,2}X^{\alpha_1+\alpha_2,1} - X^{1,1}X^{-\alpha_2,1}X^{\alpha_2,1} - \frac43 X^{1,1}X^{1,1}X^{1,1}  \nn\\&\qquad - X^{2,1}X^{-\alpha_2,1}X^{\alpha_2,1}   + 2X^{2,1}X^{1,1}X^{1,1}  - X^{2,1}X^{2,1}X^{1,1} + \frac16 X^{2,1}X^{2,1}X^{2,1} +\dots )D_{-\alpha_1,4}\nn\\
&+(X^{-\alpha_2,2}X^{\alpha_1,1} + X^{1,1}X^{-\alpha_2,1}X^{\alpha_1,1}   + X^{2,1}X^{-\alpha_2,1}X^{\alpha_1,1}+\dots)D_{-\alpha_2,4}\nn\\
&+( X^{\alpha_2,1}X^{\alpha_1,1}X^{\alpha_1,1} + X^{\alpha_1+\alpha_2,1}X^{\alpha_1,2} - X^{-\alpha_2,1}X^{\alpha_1+\alpha_2,1}X^{\alpha_1+\alpha_2,1}  + 2X^{1,1}X^{\alpha_1+\alpha_2,1}X^{\alpha_1,1} \nn\\&\qquad- X^{2,1}X^{\alpha_1+\alpha_2,1}X^{\alpha_1,1}  +\dots)D_{\alpha_1+\alpha_2,4}\nn\\
&+( - X^{-\alpha_1-\alpha_2,2}X^{\alpha_1,1}  + \frac12 X^{1,1}X^{1,1}X^{-\alpha_2,1}  + X^{2,1}X^{1,1}X^{-\alpha_2,1}  + \frac12 X^{2,1}X^{2,1}X^{-\alpha_2,1}+ \dots )D_{-\alpha_1-\alpha_2,4} \nn\\
&+(- X^{1,1}X^{-\alpha_2,1}X^{\alpha_1+\alpha_2,1}  + 2X^{1,1}X^{1,1}X^{\alpha_1,1}  - X^{2,1}X^{-\alpha_2,1}X^{\alpha_1+\alpha_2,1} \nn\\&\qquad - 2X^{2,1}X^{1,1}X^{\alpha_1,1}  + \frac12 X^{2,1}X^{2,1}X^{\alpha_1,1}+\dots )D_{1,4}\nn\\
&+(- X^{-\alpha_2,1}X^{\alpha_2,1}X^{\alpha_1,1}      - X^{-\alpha_2,2}X^{\alpha_1+\alpha_2,1} - X^{1,1}X^{-\alpha_2,1}X^{\alpha_1+\alpha_2,1}   \nn\\&\qquad - X^{2,1}X^{-\alpha_2,1}X^{\alpha_1+\alpha_2,1}+\dots )D_{2,4}+\dots
\end{align}

\bibliographystyle{alpha}
\bibliography{biblio.bib}

\end{document}